\documentclass[11pt,letterpaper]{article}

\usepackage[margin=1in]{geometry}
\usepackage{amsmath,amssymb,amsthm,mathtools}
\usepackage{microtype}
\usepackage{graphicx}
\usepackage{booktabs}
\usepackage{aliascnt}
\usepackage{xcolor}
\usepackage[
  colorlinks=true,
  linkcolor=teal,
  citecolor=teal,
  urlcolor=teal
]{hyperref}
\usepackage{cleveref}
\usepackage{enumitem}
\usepackage{authblk}

\newcommand{\E}{\mathbb E}
\newcommand{\D}{\mathcal{D}}
\newcommand{\ind}{\mathbf{1}}
\newcommand{\ALG}{\mathsf{ALG}}
\newcommand{\OPT}{\mathsf{OPT}}
\newcommand{\CR}{\mathsf{CR}}

\newcommand{\eqdef}{\vcentcolon=}
\newcommand{\eps}{\varepsilon}

\newcommand{\K}{\mathcal K}
\newcommand{\C}{\mathcal C}
\newcommand{\I}{\mathcal I}
\newcommand{\F}{\mathcal{F}}

\newif\ifshowcomments
\showcommentstrue          

\ifshowcomments
  \newcommand{\comment}[1]{\textcolor{red}{[#1]}}
  \newcommand{\daniel}[1]{\textcolor{red}{[Daniel: #1]}}
\else
  \newcommand{\comment}[1]{}
  \newcommand{\daniel}[1]{}
\fi

\title{Prophet Inequalities Beyond Utilitarian Social Welfare}

\author[1]{Daniel Halpern}
\author[2]{Abhiram Manohara}
\author[2]{Alexandros Psomas}

\affil[1]{Google Research}
\affil[2]{Purdue University}

\date{}

\newtheorem{theorem}{Theorem}[section]
\newaliascnt{lemma}{theorem}
\newtheorem{lemma}[lemma]{Lemma}
\aliascntresetthe{lemma}
\newaliascnt{proposition}{theorem}
\newtheorem{proposition}[proposition]{Proposition}
\aliascntresetthe{proposition}
\newaliascnt{corollary}{theorem}
\newtheorem{corollary}[corollary]{Corollary}
\aliascntresetthe{corollary}
\theoremstyle{definition}
\newaliascnt{definition}{theorem}

\aliascntresetthe{definition}
\newaliascnt{example}{theorem}

\aliascntresetthe{example}
\theoremstyle{remark}
\newaliascnt{remark}{theorem}

\aliascntresetthe{remark}

\begin{document}

\maketitle

\begin{abstract}
In the classical i.i.d. prophet-inequality problem, a single item must be allocated to one of $n$ agents who arrive sequentially, with values drawn independently from a known distribution. When an agent arrives, their value is revealed, and the algorithm must immediately decide whether to allocate the item or continue. The usual objective is to maximize utilitarian welfare: the expected value of the recipient. Guarantees for this objective immediately extend to the problem of allocating $m$ indivisible items to sequentially arriving agents with i.i.d.\ additive values. Utilitarian welfare, however, ignores how expected utility is distributed across agents. Motivated by a rich literature in fair division, we instead evaluate an online rule by its generalized $p$-mean welfare. This family includes utilitarian welfare at $p=1$, Nash welfare (the geometric mean of utilities) at $p=0$, and egalitarian welfare (the minimum utility) as $p$ tends to $-\infty$.

When the number of items is large, we show that our many-item fair-division problem is captured exactly by a single-item prophet problem in which an online rule is evaluated by the $p$-mean of the agents' expected utilities. We then characterize this single-item problem; every online rule is weakly Pareto dominated by a quantile-threshold rule, and an optimal egalitarian rule equalizes the agents' expected utilities. We prove that, for every $n$, the online optimum is at least $\Gamma \approx 0.7059$ times the prophet's egalitarian welfare; by monotonicity of the generalized means, the same guarantee extends to every $p\le1$ in the single-item problem. We further prove that, for egalitarian welfare, the optimal ratio converges to $\Gamma$ as $n\to\infty$. Thus, asymptotically, optimizing egalitarian rather than utilitarian welfare costs only about four percentage points, relative to the classical utilitarian ratio of $0.7451$. Finally, returning to fair division, we prove that the large-item assumption is necessary: when $m=n$, the worst-case competitive ratio converges to zero as $n \to \infty$ for every $p\le0$.

\end{abstract}

\section{Introduction}
\label{sec:introduction}

Consider the following online allocation problem. There are $m$ indivisible items and $n$ agents with additive preferences. Agents arrive sequentially; when an agent arrives, their values for the remaining items are revealed, and we must irrevocably choose a bundle for them. Item values are drawn i.i.d.\ from a known distribution. When our objective is to maximize utilitarian social welfare (the sum of agents' utilities) this problem decomposes item by item. Each item gives rise to a classical i.i.d.\ prophet-inequality problem: a single item must be allocated to one of $n$ sequentially arriving agents, and the algorithm must decide whether to allocate after observing each value. Running a prophet algorithm independently on every item therefore gives the same competitive guarantee in the many-item problem. The optimal ratio for the single-item problem converges to approximately $0.7451$ as the number of agents grows~\cite{hillkertz1982,kertz1986,correa2017posted}.

Utilitarian welfare, however, ignores how utility is distributed across agents: an allocation may generate high total value, but leave some agents with very little. A rich literature on fair division asks how to balance efficiency with meaningful utility guarantees for individual agents. A standard way to capture this tradeoff is generalized $p$-mean welfare~\cite{moulin2003fair,barman2020tight,barman2022universal}. This family includes utilitarian welfare at $p=1$ (the average utility), \emph{Nash} welfare (the geometric mean of agents' utilities) at $p=0$, and egalitarian welfare (the minimum utility across agents) at $p=-\infty$; as $p$ decreases, the objective places increasing weight on agents with lower utility.

The picture becomes particularly clean when the number of items is large. Fixing $n$ and the value distribution, independently applying a single-item rule to every item makes each agent’s utility per item converge to its expected utility under that rule as the number of items grows large. We show that, in this limit, our fair-division problem is captured exactly by a single-item prophet problem in which an online rule is evaluated by the $p$-mean of the agents' expected utilities.

For the prophet, fairness over expected utilities comes at no cost: allocating to an agent of maximum value and breaking ties uniformly maximizes total value while giving every agent the same expected utility. Online, however, these goals can conflict. Allocating more readily benefits early agents, while waiting preserves opportunities for later agents. Our main question is how much of the prophet's welfare an online rule can retain when the objective places greater weight on agents with lower expected utility.

\subsection{Our Contributions}

For a distribution $\D$, let $\ALG_n(p;\D)$ denote the optimal online
$p$-mean welfare. As shown in \Cref{sec: model}, the prophet's optimal
value is the same for every $p \in [-\infty, 1]$; we denote this common value by
$\OPT_n(\D)$.
We study the worst-case competitive ratio
\[
  \CR_n(p)
  =
  \inf_{\D}
  \frac{\ALG_n(p;\D)}{\OPT_n(\D)}.
\]

Our main result gives a tight answer for egalitarian welfare as the number of
agents grows. Let
\(
  \Gamma
  \approx0.70591334.
\)
We prove that for every $n \ge 2$, $
  \CR_n(-\infty)\ge\Gamma$,
and that this guarantee becomes tight as the number of agents grows: 
  $\lim_{n\to\infty}\CR_n(-\infty)=\Gamma$ (\Cref{thm:main}).
Thus, even the least-favored agent can be guaranteed more than $70\%$ of their
share of the prophet's value. Compared with the classical asymptotic ratio of
approximately $0.7451$, the price of requiring an equal guarantee across
arrival positions is only about four percentage points.

This result also gives guarantees across the entire spectrum of $p$-mean
objectives. The generalized mean of a fixed utility profile is nondecreasing in
$p$, while the offline benchmark (the prophet's performance) is the same for all $p \in [-\infty, 1]$. Hence $\CR_n(p)$ is
nondecreasing in $p$, and we can conclude that $\CR_n(p)\ge\Gamma$ for every $n\ge2$ and every
$p\in[-\infty,1]$. In fact, the same optimal egalitarian rule achieves this guarantee simultaneously for every $p \in [-\infty, 1]$. Together with the classical upper bound at $p=1$, this gives,
for every fixed $p\in[-\infty,1]$,
\[
  0.7059 \approx \Gamma
  \le \liminf_{n\to\infty}\CR_n(p)
  \le \limsup_{n\to\infty}\CR_n(p)
  \le \beta_{\mathrm{iid}}
  \approx 0.7451,
\]
where $\beta_{\mathrm{iid}}=\lim_{n\to\infty}\CR_n(1)$ is the classical
i.i.d.\ prophet constant~\cite{hillkertz1982,kertz1986,correa2017posted}.

Returning to the many-item fair division problem, we show that applying the $p$-mean to expected utilities gives an exact reduction to the single-item problem: for every $m$, the optimal online and offline welfares are $m$ times their single-item counterparts, and hence the competitive ratios coincide (\Cref{prop:fair-division-ex-ante}). When the $p$-mean is instead applied to realized (ex-post) utilities, the same connection holds asymptotically: for every fixed $n$, $p \in [-\infty, 1]$, and distribution $\D$, the competitive ratio converges to the corresponding single-item ratio as $m\to\infty$ (\Cref{prop:fair-division-limit}). Finally, we prove that this large-item
assumption is necessary: when $m=n$, the worst-case ex-post ratio tends to zero for every $p\le0$ (\Cref{prop:fair-division-impossibility}).

\subsection{Related Work}

Our work lies at the intersection of prophet inequalities and fair division. The classical prophet inequality, originating with Krengel and Sucheston~\cite{krengelsucheston1977}, guarantees one half of the prophet's expected value for independent, nonnegative random variables; this guarantee can be achieved by a single-threshold rule~\cite{samuelcahn1984}.  The i.i.d.\ case admits a strictly better constant.  Hill and Kertz characterized the finite-horizon extremal problem~\cite{hillkertz1982}, Kertz identified its asymptotic value~\cite{kertz1986}, and Correa et al.\ gave a matching algorithm, establishing the tight asymptotic ratio of approximately $0.7451$~\cite{correa2017posted}. This line of work asks how much total value an online rule can guarantee.  We instead ask what guarantees are possible for the entire vector of expected
utilities across arrival positions.

Our work is closest to Hill and Kennedy~\cite{hillkennedy1990}, who developed a framework for optimal stopping with generalized objective functions.  In particular, they explicitly considered $\max_\tau\min_j\E[X_j\mathbf 1\{\tau=j\}]$, exactly our egalitarian objective, for general sequences that need not be identically distributed.  Their work establishes universal prophet bounds for this generalized objective, and the ordinary prophet inequality, applied to a weighted scalarization of the utility vector, gives a $1/2$ guarantee in the general independent setting.  Our work focuses on what the additional i.i.d.\ structure buys: it leads to a strictly larger sharp asymptotic constant and gives guarantees for the entire family of $p$-mean objectives.  Arsenis and Kleinberg~\cite{arseniskleinberg2022} pursue a different notion of fairness, requiring identity-independent and time-independent treatment of candidates, and recover a $1/2$ prophet inequality even when both constraints are imposed.  Thus, their fairness requirements constrain the stopping rule, whereas ours directly evaluates the resulting vector of agents' utilities.

This objective-based view of fairness is motivated by a rich literature on fair division, where generalized $p$-mean welfare is used to interpolate between utilitarian, Nash, and egalitarian welfare~\cite{moulin2003fair,barman2020tight,barman2022universal}. Many works on online fair division fix the set of agents and allocate items that arrive online, focusing on various objectives, e.g., minimizing maximum envy~\cite{benade2024fair,benade2025dynamic,halpern2025online,mertzanidis2024automating,gkatzelis2021fair,he2019achieving,procaccia2024honor,schiffer2025improved,amanatidis2025online,kulkarni2026online}, maximizing Nash welfare~\cite{banerjee2022online,gao2021online,liao2022nonstationary,yang2024greedy,huang2023online} and egalitarian welfare~\cite{kawase2022online,springer2022online}, or, more generally, maximizing the generalized $p$-mean of utilities~\cite{huang2025long,barman2022universal,yang2026online,balseiro2023uniformly}. Our model reverses the direction of arrival: i.i.d.\ agents arrive against a fixed collection of items. Online arrival of agents has been primarily studied for (static) divisible items/resources~\cite{kash2014no,li2018dynamic,im2020dynamic,vardi2022dynamic,sinclair2023sequential,banerjee2026online}. The closest work to ours is that of Kulkarni et al.~\cite{kulkarni2025towards}, who study ex-post maximin-share guarantees for a fixed set of indivisible items when agents from a known set of additive valuation types arrive adversarially or stochastically. Finally, a related literature fixes the agents and studies repeated resource allocation under dynamic demands or resource availability~\cite{vuppalapati2023karma,fikioris2024incentives,banerjee2023robust,fikioris2025beyond}.

\subsection{Organization.}
\Cref{sec: model} defines the single-item model.
\Cref{sec: results} states our main results for this model and
characterizes the structure of optimal online rules.
\Cref{sec:fair-division} develops the connection to online fair
division. Finally, \Cref{sec:proof} presents the main proof ideas and
proves our central result. The appendices contain the remaining proofs
and technical details.

\section{Model}\label{sec: model}

Throughout the paper, all value distributions are supported on $\mathbb R_{\ge0}$ and have finite, positive mean. Fix an integer $n\ge2$, write $[n]\eqdef{1,\dots,n}$, and let $\D$ be such a distribution with CDF $F$.
The agents' values $X_1,\dots,X_n$ are drawn independently from $\D$.

There is one item. The agents arrive in the order $1,\dots,n$. When
agent $i$ arrives, we observe $X_i$ and must immediately decide whether
to allocate the item to $i$ or continue. The item may be allocated at
most once.

\paragraph{Allocation Rules and Feasible Utilities.}

An \emph{allocation rule} $\pi$ is a random variable taking values in
$[n]\cup\{\varnothing\}$, where $\pi=i$ means that agent $i$
receives the item and $\pi=\varnothing$ means that the item is never
allocated. A rule is \emph{online} if its decision at step $i$ may use
only $X_1,\dots,X_i$ and internal randomness. An \emph{offline} rule
(or a \emph{prophet}) may depend on the entire vector
$(X_1,\dots,X_n)$.

Agent $i$'s expected utility under $\pi$ is
\[
  v_i(\pi,\D)\eqdef\E\bigl[X_i\ind\{\pi=i\}\bigr].
\]
When $\D$ is clear from context, we abbreviate this to $v_i(\pi)$. The resulting utility profile is
$\mathbf v(\pi,\D)=(v_1(\pi,\D),\dots,v_n(\pi,\D))$. Define
\begin{align*}
  \mathcal V^{\mathrm{on}}(\D)
  &\eqdef
  \{\mathbf v(\pi,\D):\text{$\pi$ is an online rule}\},\\
  \mathcal V^{\mathrm{off}}(\D)
  &\eqdef
  \{\mathbf v(\pi,\D):\text{$\pi$ is an offline rule}\}.
\end{align*}

\paragraph{Social Welfare Objectives.}

For finite $p\le1$ with $p\ne0$, define the generalized $p$-mean
welfare of a utility profile $\mathbf v$ by
\[
  w_p(\mathbf v)
  \eqdef
  \left(\frac1n\sum_{i=1}^n v_i^p\right)^{1/p},
\]
where this value is taken to be $0$ for $p \le 0$ when any $v_i = 0$. 
When $p=1$, this is utilitarian social welfare, measured as average
expected utility: $w_1(\mathbf v)=\frac1n\sum_{i=1}^n v_i$.
It can also be continuously extended to $p = 0$ and $p = -\infty$. The limiting cases are Nash welfare (geometric mean),
$w_0(\mathbf v)=(\prod_{i=1}^n v_i)^{1/n}$, and egalitarian welfare
(minimum utility), $w_{-\infty}(\mathbf v)=\min_{i\in[n]}v_i$.
For every $p \in [-\infty, 1]$, $w_p$ is coordinatewise nondecreasing, concave, and positively homogeneous: $w_p(\lambda\mathbf v)=\lambda w_p(\mathbf v)$ for $\lambda\ge0$. Moreover, the power-mean inequality gives $w_{p_1}(\mathbf v)\ge w_{p_2}(\mathbf v)$ whenever $p_1\ge p_2$.

\paragraph{Online and Offline Optima.}
For each $p \in [-\infty, 1]$, define the optimal online value
\[
  \ALG_n(p;\D)
  \eqdef
  \sup_{\mathbf v\in\mathcal V^{\mathrm{on}}(\D)}
  w_p(\mathbf v).
\]
Observe that the offline optimum is the same for every $p \in [-\infty, 1]$: for any offline rule, the power-mean inequality gives
\[
  w_p(\mathbf v)
  \le w_1(\mathbf v)
  =\frac1n\sum_{i=1}^n v_i \le\frac1n\E\left[\max_{i\in[n]}X_i\right].
\]
Conversely, allocating to an agent of maximum value and breaking ties uniformly at random gives every agent expected utility $\E[\max_{i \in [n]} X_i]/n$, achieving the previous upper bound. We denote this common offline value by 
\[
   \OPT_n(\D) \eqdef  \sup_{\mathbf v\in\mathcal V^{\mathrm{off}}(\D)}
  w_p(\mathbf v)
  =
 \frac1n\E\left[\max_{i\in[n]}X_i\right]
  \qquad\text{for every }p \in [-\infty, 1].
\]
Thus the worst-case competitive ratio is
\[
  \CR_n(p)
  \eqdef
  \inf_{\D}
  \frac{\ALG_n(p;\D)}{\OPT_n(\D)}.
\]
Here and throughout, the infimum ranges over distributions on
$\mathbb R_{\ge0}$ with finite, positive mean.
When $p=1$, $\CR_n(1)$ reduces to the classic i.i.d.\ prophet inequality, where it is known that $\lim_{n\to\infty}\CR_n(1) \approx0.7451$~\cite{correa2017posted}.

\paragraph{Quantile Representation and Threshold Rules.}

It is convenient to parameterize values by their quantiles. Let
$F^{-1}(q) = \inf\{t \in \mathbb{R} \mid F(t) \ge q\}$ be the quantile
function of $\D$. We may assume without loss of generality that the
online algorithm observes i.i.d.\ uniform quantiles
$U_1, \dots, U_n \sim \operatorname{Unif}[0, 1]$, where
$X_i = F^{-1}(U_i)$ almost surely.\footnote{This introduces no extra
informational power: upon observing $X_i$, the algorithm can draw an
independent uniform variable $R_i \sim \operatorname{Unif}[0, 1]$ and set
$U_i = F(X_i^-) + R_i(F(X_i) - F(X_i^-))$, where
$F(t^-) \eqdef \lim_{s \to t^-} F(s)$. The resulting variables
$U_i$ are i.i.d.\ $\operatorname{Unif}[0,1]$ and satisfy
$F^{-1}(U_i) = X_i$ almost surely.}

A \emph{quantile-threshold rule} is specified by thresholds
$q_1,\ldots,q_{n-1}\in[0,1]$. At step $i<n$, it allocates the item to
agent $i$ if the item is still available and $U_i\ge q_i$. If the item
remains available at step $n$, it is allocated to agent $n$. For
notational convenience, we may set $q_n=0$, so that the same description
applies to every step.

For a quantile-threshold rule with threshold vector
$\mathbf q=(q_1,\ldots,q_{n-1})$, let $r_i(\mathbf q)$ denote the
probability that the item reaches agent $i$. Then
$r_1(\mathbf q)=1$ and
$r_{i+1}(\mathbf q)=r_i(\mathbf q)q_i$. Equivalently,
$r_i(\mathbf q)=\prod_{j=1}^{i-1}q_j$. When the threshold rule is
clear from context, we abbreviate $r_i(\mathbf q)$ to $r_i$.

For a distribution $\D$ with CDF $F$,  define
$S_\D(q)\eqdef\E[F^{-1}(U)\ind\{U\ge q\}] = \int_q^1F^{-1}(u) \, du$ where $U\sim\operatorname{Unif}[0,1]$. When $\D$ is clear from context,
we abbreviate $S_\D$ to $S$. Conditional on the item reaching an agent,
threshold $q$ gives that agent expected utility $S(q)$ and
leaves the item available for the next agent with probability $q$. Consequently, for a quantile-threshold rule
$\pi$,  $v_i(\pi, \D) = r_i S_\D(q_i)$.

\section{Results}\label{sec: results}

For \(a,b>0\), define
\[
    \rho(a,b)
    \eqdef
    \frac{1}{
        \left(1-\dfrac{b}{a+b}e^{-a}\right)
        \left(1+\dfrac{1-e^{-b}}a\right)
    }.
\]
and let
\(
  \Gamma
  \eqdef
  \inf_{a,b>0}\rho(a,b).
\)
Numerical minimization gives $\Gamma\approx0.70591334$, with a minimizer near $a\approx0.81833$ and $b\approx2.42237$. Our main theorem states that the worst-case competitive ratio for egalitarian welfare is bounded below by \(\Gamma\) for every \(n\) and converges to \(\Gamma\) as \(n\to\infty\).

\begin{theorem}\label{thm:main}
For every \(n\ge2\), $\CR_n(-\infty)\ge\Gamma$. Moreover, $\lim_{n\to\infty}\CR_n(-\infty)=\Gamma$.
\end{theorem}

\Cref{sec:proof} is devoted to the proof of \Cref{thm:main}. It begins
with a detailed overview of the main ideas, followed by proofs of the supporting lemmas and the theorem itself.

Because $w_p(\mathbf v)$ is nondecreasing in $p$, both $\ALG_n(p;\D)$ and $\CR_n(p)$ are nondecreasing in $p$. The following corollary is immediate.

\begin{corollary}\label{cor:all-p-lower-bound}
For every $n\ge 2$ and every $p\in[-\infty,1]$, $\Gamma \le \CR_n(p)$, and
\[
  0.7059 \approx \Gamma
  \le \liminf_{n\to\infty}\CR_n(p)
  \le \limsup_{n\to\infty}\CR_n(p)
  \le \beta_{\mathrm{iid}}
  \approx 0.7451.
\]
\end{corollary}

Finally, it is worth noting that for two agents, the egalitarian competitive ratio coincides with the classical utilitarian ratio~\cite{hillkertz1982}. By monotonicity, the same is then true for every intermediate $p$-mean objective. However, this coincidence already disappears with three agents. The proof, given in
Appendix~\ref{subsec: proof of gap for 3 agents}, combines a direct
two-agent analysis with a three-agent instantiation of the limiting hard distribution from \Cref{thm:main}.

\begin{proposition}\label{prop:small-n}
    For every $p \in [-\infty, 1]$,
    \(
        \CR_2(p) = \frac{2 + \sqrt{2}}{4} \approx 0.8536.
    \)
    By contrast,
    \( \CR_3(-\infty) \leq \frac{128}{159} \approx 0.805 < 0.818 \approx \CR_3(1). \)
\end{proposition} 

\subsection{Structure of Optimal Online Rules}

Our proof relies on two structural results about optimal rules. The
first, proved in Appendix~\ref{app:pareto}, shows that
quantile-threshold rules dominate arbitrary online rules and
characterize the Pareto frontier of online expected-utility profiles. This is analogous to the utilitarian setting~\cite{hillkertz1982}.

\begin{proposition}\label{prop:pareto}
    For every distribution $\mathcal{D}$ with CDF $F$, the following hold:
    \begin{enumerate}[leftmargin=*]
        \item For every online rule $\pi$, there is a threshold rule $\pi'$ such that $v_i(\pi', \mathcal{D}) \ge v_i(\pi, \mathcal{D})$ for all $i \in [n]$.
        \item For every threshold rule $\pi$ parameterized by $q_1, \ldots, q_{n-1}$ such that $q_i \ge F(0)=\Pr[X=0]$ for all $i \in [n-1]$, $\mathbf{v}(\pi, \mathcal{D})$ is Pareto optimal, i.e., there is no $\mathbf{u} \in \mathcal{V}^{\mathrm{on}}(\mathcal{D})$ such that $u_i \ge v_i(\pi, \mathcal{D})$ for all $i \in [n]$, with at least one inequality strict.
    \end{enumerate}
\end{proposition}

Second, as in the utilitarian setting, an optimal rule can be
constructed by working backward from the final agent. For utilitarian
welfare, the final agent receives the item whenever it remains
available, and agent $i$ receives it whenever their value exceeds the
expected value obtained by continuing with agents $i+1,\ldots,n$.
Our egalitarian construction is analogous: we choose the threshold for
agent $i$ so that, conditional on the item reaching that agent, their
expected utility equals that of the final agent. Proceeding backward
in this way produces an optimal rule that equalizes the agents'
expected utilities. 
\begin{proposition}\label{prop:equalizing-rule}
Fix a distribution $\D$ with CDF $F$ and let $\mu=\E[X]$. There is a sequence
of thresholds $q_1,\ldots,q_{n-1}\in(F(0),1)$ satisfying
\[
\frac{S(q_i)}{q_i}
=
\mu\prod_{j=i+1}^{n-1}q_j,
\]
where the empty product is $1$. Let $\pi^\star$ be the corresponding
quantile-threshold rule, and define its reach probabilities by
$r_1=1$ and $r_{i+1}=r_iq_i$. Then $\pi^\star$ is optimal for
egalitarian welfare and gives every agent expected utility
\(
v_i(\pi^\star,\D)
=
\mu r_n
\).
Consequently,
\(
\ALG_n(-\infty;\D)
=
\mu r_n.
\)
\end{proposition}

We prove \Cref{prop:equalizing-rule} at the beginning of
\Cref{sec:proof}, since its recurrence is a key ingredient in the proof
of our main result.

\section{Connection to Online Fair Division}\label{sec:fair-division}

The fairness objective in our prophet problem measures how expected
utility is distributed across agents. Here, we show how the same
objective arises in an online fair-division problem with many indivisible items and sequentially arriving agents.

There are \(n\) agents and \(m\) indivisible items. Agent \(i\)'s value
for item \(j\) is \(X_{ij}\), and all values are drawn independently
from a distribution \(\D\). The agents arrive one at a time. When an
agent arrives, their values for the remaining items are revealed, and
an online algorithm must immediately decide which of these items to
allocate to them. An offline algorithm, by contrast, sees the entire
valuation matrix before choosing an allocation. If agent \(i\)
receives the items \(A_i \subseteq [m]\), their realized utility is
\(U_i=\sum_{j\in A_i}X_{ij}\).

There are two natural ways to evaluate welfare in this model. The
\emph{ex-ante} (ea) objective applies the \(p\)-mean to the agents'
expected utilities, giving
\(w_p(\E[U_1],\ldots,\E[U_n])\). The \emph{ex-post} (ep) objective instead
takes the expected \(p\)-mean of the realized utilities, giving
\(\E[w_p(U_1,\ldots,U_n)]\). For egalitarian welfare, these objectives
are \(\min_i\E[U_i]\) and \(\E[\min_iU_i]\), respectively.

For \(s\in\{\mathrm{ea},\mathrm{ep}\}\), let
\(\ALG^s_{n,m}(p;\D)\) and \(\OPT^s_{n,m}(p;\D)\) denote the optimal
online and offline values under objective \(s\), and let $\CR^s_{n, m}(p)$ be the corresponding competitive ratio.\footnote{For ex-post competitive ratios with $p\le0$, we restrict to $m\ge n$, since both optima are zero when $m<n$.} We first observe that the ex-ante problem is equivalent to the single-item problem, for every number of items. Proofs for this section can be found in Appendix~\ref{app:fair-division}.

\begin{proposition}
\label{prop:fair-division-ex-ante}
For every $n \ge 2$, $m \ge1$ and distribution $\D$, the sets of expected-utility profiles achievable by online and offline rules in the \(m\)-item problem are, respectively, exactly \(m\) times the corresponding sets in the single-item problem. Thus, for all $p \in [-\infty, 1]$,
$
    \ALG^{\mathrm{ea}}_{n,m}(p;\D)=m\ALG_n(p;\D)$ and 
    $\OPT^{\mathrm{ea}}_{n,m}(p;\D)=m\OPT_n(\D)$. Consequently, $\CR^{\mathrm{ea}}_{n, m}(p) = \CR_n(p)$. 
\end{proposition}

One direction follows by applying a single-item rule independently to
every item. For the other, start with an arbitrary many-item rule and
focus on a uniformly random item. By internally simulating the values
of all other items, we obtain a single-item rule whose expected-utility
profile is the average profile of the original rule. 

The same equivalence does not hold for ex-post welfare when the number of items is small. In fact, when there is exactly one item per agent, the prophet
may give every agent positive utility with high probability even
though an online algorithm has almost no chance of doing so.

\begin{proposition}
\label{prop:fair-division-impossibility}
For every \(p\le0\), the worst-case ex-post competitive ratio with
\(m=n\) converges to zero as \(n\) grows. More specifically, \(\CR^{\mathrm{ep}}_{n, n}(p) \in O(\log n/n)\).
\end{proposition}

The construction takes
\(\D=\operatorname{Bernoulli}(2\log(n)/n)\). The prophet can give every
agent a value-one item whenever the corresponding random bipartite
graph contains a perfect matching, which happens with probability
approaching one~\cite{frieze2015introduction}. By contrast, for an
online allocation to have positive \(p\)-mean welfare when \(p\le0\),
every agent must receive a value-one item. The final agent must
therefore value the one remaining item, an event of probability only
\(2\log(n)/n\).

This impossibility relies on the numbers of agents and items growing
together. If we instead fix the number of agents and the distribution,
ex-post welfare approaches ex-ante welfare as the number of items
grows.

\begin{proposition}
\label{prop:fair-division-limit}
For every \(n \ge 2\), \(p\in [-\infty, 1]\), and distribution \(\D\),
\[
    \lim_{m\to\infty}
    \frac{\ALG^{\mathrm{ep}}_{n,m}(p;\D)}
         {\OPT^{\mathrm{ep}}_{n,m}(p;\D)}
    =
    \frac{\ALG_n(p;\D)}{\OPT_n(\D)}.
\]

\end{proposition} 

The lower bound is obtained by applying an optimal single-item
rule  independently to every item. The law of large numbers then makes
each agent's average realized utility approach their expected utility
under that rule. In the other direction, concavity of the \(p\)-mean
bounds ex-post welfare by ex-ante welfare, and
\Cref{prop:fair-division-ex-ante} applies. A finite mean is enough for
this pointwise convergence, although it does not provide a uniform
rate over distributions.

Stronger quantitative guarantees follow when the distribution has lighter tails using concentration and moment bounds. For example, 
suppose \(X\le B \cdot \E[X]\) almost surely, then for $m \in \Omega\left(\frac{nB}{\delta^2} \log\left( \frac{n}{\delta}\right)\right)$, $\frac{\ALG^{\mathrm{ep}}_{n,m}(p;\D)}
         {\OPT^{\mathrm{ep}}_{n,m}(p;\D)}\ge \frac{\ALG_{n}(p;\D)}
         {\OPT_{n}(\D)}-\delta$.
Under the weaker assumption of a finite mean $\mu$ and finite variance $\sigma^2$, we can instead require $m \in \Omega\left( \frac{n^3}{\delta^2}\left( 1 + \frac{\sigma^2}{\mu^2} \right) \right)$. Details can be found in Appendix~\ref{app:bounds}.

These results explain why we focus on the single-item problem: this problem captures ex-ante fair division exactly for every number of items and captures ex-post fair division in the many-item limit. At the same time, the failure when \(m=n\) shows that this convergence is not uniform when $n$ and $m$ grow together. 

\section{Proof of Main Result}\label{sec:proof}
\subsection{Proof Overview}

We begin with the lower bound. Scaling all values by a positive constant scales both $\ALG_n(-\infty;\D)$ and $\OPT_n(\D)$ by that constant, and therefore does not change their ratio. Hence, it suffices to prove the lower bound for distributions satisfying $\E[X]=1$; we make this assumption throughout the lower-bound argument.

The lower bound has two main ingredients. First, we transform each
distribution into a convex function whose objective value is exactly
the prophet-to-online ratio. Second, we solve the resulting variational
optimization problem.

Let $\K_{\mathrm{fin}}$ denote the class of nonnegative,
nondecreasing, convex functions
$G\colon[0,\infty)\to[0,\infty]$ such that $G(0)=0$ and
$G(z)=+\infty$ for all sufficiently large $z$. For
$G\in\K_{\mathrm{fin}}$, define
\[
\I(G)\eqdef\int_0^\infty e^{-G(z)}\,dz,
\qquad
\C(G)\eqdef\int_1^\infty\frac{dz}{zG(z)},
\]
with the conventions $e^{-\infty}=0$, $1/(+\infty)=0$, and
$1/0=+\infty$.

For a distribution $\D$ with $\mu=\E[X]$, recall that
$S(q)=\E[F^{-1}(U)\ind{U\ge q}]$. For $0\le y\le\mu$, define its
generalized inverse by
\[
S^{-1}(y)\eqdef
\sup\{q\in[0,1]:S(q)\ge y\}.
\]

The following lemma gives the reduction from the prophet problem to
the variational problem.

\begin{lemma}
\label{lem:variational-reduction}
Fix a distribution $\D$ satisfying $\E[X]=1$, and let
$c=\ALG_n(-\infty;\D)$. Define $H\colon[0,\infty)\to[0,\infty]$ by
\[
H(z)\eqdef
\begin{cases}
-(n-1)\log\bigl(S^{-1}(cz)\bigr),
&0\le z\le 1/c,\\[1mm]
+\infty,
&z> 1/c.
\end{cases}
\]
Then $H\in\K_{\mathrm{fin}}$ and
\[
\I(H)
=
\frac{\OPT_n(\D)}{\ALG_n(-\infty;\D)}
\qquad\text{and}\qquad
\C(H)\le1.
\]
\end{lemma}

The analytic core of the proof is the following variational
inequality.

\begin{lemma}
\label{lem:variational-bound}
Every $G\in\K_{\mathrm{fin}}$ satisfying $\C(G)\le1$ also satisfies
\(
\I(G)\le\frac1\Gamma.
\)
\end{lemma}

Together, these lemmas bound the prophet-to-online ratio by $1/\Gamma$, giving $\CR_n(-\infty)\ge\Gamma$.

We next explain where the function $H$, its objective $\I(H)$, and the
constraint $\C(H)\le1$ come from. We then describe the main idea behind
the variational inequality and, finally, the distributions giving the
matching upper bound.

\paragraph{Where the variational reduction comes from.}

Recall that the pair $(S(q),q)$ describes everything that matters about using threshold $q$ at a single agent. Conditional on the item still being available, the threshold gives the current agent expected utility $S(q)$ and preserves the item for the next agent with probability $q$. Lowering $q$ gives more value to the current agent, but leaves less chance for everyone who comes later. The generalized inverse has an equally useful interpretation. If the current agent must receive conditional expected utility $y$, then $S^{-1}(y)$ is the largest threshold that provides this utility, and hence the largest probability with which the item can be preserved for the next agent. Thus, $S^{-1}(y)$ describes how much reach probability can be retained while providing a specified amount of utility now.

Let $c=\ALG_n(-\infty;\D)$. The function $S^{-1}$ is useful in two
ways. First, a direct calculation expresses the prophet-to-online ratio
as
\[
\frac{\OPT_n(\D)}{c}
=
\int_0^{1/c}
\left(S^{-1}(cz)\right)^{n-1}\,dz.
\]

Second, $S^{-1}$ describes the constraints imposed by the equalizing
rule from \Cref{prop:equalizing-rule}. Let
$q_1,\ldots,q_{n-1}$ be its thresholds, and let $r_i$ be the probability
that agent $i$ is reached. Because every agent receives
expected utility $c$, we have
$
r_iS(q_i)=c$, so
$
q_i=S^{-1}(c/r_i)$, and hence
$
r_{i+1}=r_iS^{-1}(c/r_i)$.
Moreover, $r_1=1$ and $r_n=c$, since the final agent receives the item
whenever reached and $\E[X]=1$.

Our goal is therefore to bound the prophet-to-online ratio over all
functions $S^{-1}$ whose associated reach probabilities can move from
$r_1=1$ to $r_n=c$ in $n-1$ steps.
Working with these expressions directly is awkward because both have
a multiplicative structure: the objective contains an $(n-1)$st power,
while the recurrence repeatedly multiplies the reach probability by
the next threshold. Taking logarithms simplifies both.

In the
objective, we can write
$(S^{-1}(cz))^{n-1}
=\exp\bigl((n-1)(\log S^{-1}(cz))\bigr)$.
For the recurrence, set $x_i=-\log r_i$. Thus, $x_i$ records the
logarithmic decrease in reach probability by the time agent $i$
arrives. The sequence begins at $x_1=0$ and ends at
$x_n=-\log c$, and the recurrence becomes
\[
x_{i+1}-x_i
=
-\log q_i
=
-\log S^{-1}(ce^{x_i}).
\]

It is useful to regard the points $(i,x_i)$ as a trajectory mapping arrival
positions to log reach probabilities. For our purposes, it is
more convenient to consider the inverse trajectory, which records the
arrival position at which the reach probability falls to a given
level. While this trajectory is only defined on $n$ discrete points, we can create a more continuous trajectory (more amenable to constraints) by interpolating between these points in a piecewise linear fashion.

Formally, let $T\colon[0,-\log c]\to[1,n]$ be the piecewise-linear
function satisfying $T(x_i)=i$.
On each interval
$(x_i,x_{i+1})$,
the slope (and hence derivative) is
\[
T'(x)
=
\frac{1}{x_{i+1}-x_i}
=
\frac{1}{-\log S^{-1}(ce^{x_i})}.
\]
The fundamental theorem of calculus gives a clean invariant of the trajectory that
$\int_0^{-\log c}T'(x)\,dx=T(-\log c)-T(0)=n-1$. This essentially
captures the recurrence constraint: the inverse trajectory moves from agent $1$ to agent $n$ on the interval $[0, -\log c]$, so
its total increase is exactly $n-1$.

The next key idea is that we can approximate $T'(x)$ from below by a
continuous function. Consider
$1/[-\log S^{-1}(ce^x)]$. Since $S^{-1}$ is decreasing, for
$x\in(x_i,x_{i+1})$,
\[
\frac{1}{-\log S^{-1}(ce^x)}
\le
\frac{1}{-\log S^{-1}(ce^{x_i})}
=
T'(x).
\]
This bound is close to equality when consecutive reach probabilities
are close, which helps explain why the resulting continuous relaxation
can be asymptotically tight.

Integrating gives
\[
\int_0^{-\log c}
\frac{dx}{-\log S^{-1}(ce^x)}
\le \int_0^{-\log c} T'(x) \, dx = n-1.
\]
After the change of variables $z=e^x$, this becomes
\[
\int_1^{1/c}
\frac{dz}{-z\log S^{-1}(cz)}
\le n-1.
\]

At this point, the same expression $\log S^{-1}(cz)$ appears in both
the prophet's value and the constraint. In the objective, it is
multiplied by $n-1$ in the exponent; in the constraint, the right-hand
side is $n-1$. This motivates defining
\[
H(z)\coloneqq
\begin{cases}
-(n-1)\log S^{-1}(cz),&0\le z\le1/c,\\
+\infty,&z>1/c.
\end{cases}
\]
The prophet-to-online ratio and the constraint now take the simple
forms
\[
\I(H)
\coloneqq
\int_0^\infty e^{-H(z)}\,dz
=
\frac{\OPT_n(\D)}{c},
\qquad
\C(H)
\coloneqq
\int_1^\infty\frac{dz}{zH(z)}
\le1.
\]
To reemphasize, in the formal proof, we bypass the interpolated trajectory: we define
$H$ directly and prove $\C(H)\le1$ by summing over the intervals
$[1/r_i,1/r_{i+1}]$. The interpolated trajectory viewpoint is included only to
motivate where this particular integral constraint arises.

\paragraph{The variational inequality.}

The proof of \Cref{lem:variational-bound} primarily shows that any maximizer must take
the form
\[
G^\star(z)=
\begin{cases}
0,&0\le z\le J,\\
\sigma(z-J),&J<z<M,\\
+\infty,&z\ge M,
\end{cases}
\]
for some $0\le J<1<M$ and $\sigma>0$. The key observation is that any
such convex function can be represented as a mixture of
\emph{hinges}, functions of the form $h_a(z)=(z-a)_+$. We examine the first-order effect of adding a hinge and then rescaling to preserve the constraint. At a maximizer, each such perturbation has nonpositive first-order effect. However, the combined effect of the hinges already present in $G^\star$ is zero: adding them in their existing proportions simply rescales $G^\star$, which the normalization undoes. Thus, every hinge in the support of the representing measure must have zero first-order effect. Analyzing these conditions shows that the measure is supported only at $J$, giving the form above.

The expression for $\rho(a,b)$ arises from evaluating $\I(G^\star)$.
Writing $a=\sigma J$ and $b=\sigma(M-J)$ gives
$\I(G^\star)=1/\rho(a,b)$. Taking the infimum over $a,b>0$ therefore
gives $\I(G^\star)\le1/\Gamma$ and establishes the desired lower bound.

\paragraph{The matching upper bound.}

To construct a hard distribution, for arbitrary $a,b>0$, we consider
a distribution that takes value $1/\eps$ with probability $\eps$,
value $n/a$ with probability $b/n$, and value $0$ otherwise. The rare,
very large value and the more common intermediate value produce the
two factors appearing in $\rho(a,b)$.

A direct analysis of the equalizing threshold recurrence bounds the
online value, while the prophet's value can be calculated explicitly.
After first taking $\eps\downarrow0$ and then letting $n\to\infty$, the
resulting competitive ratio is at most
\[
\left[
\left(1-\frac{b}{a+b}e^{-a}\right)
\left(1+\frac{1-e^{-b}}a\right)
\right]^{-1}
=
\rho(a,b).
\]
Since $a,b>0$ are arbitrary, this gives
$\limsup_{n\to\infty}\CR_n(-\infty)\le\Gamma$. Together with the
variational lower bound, it follows that
$\lim_{n\to\infty}\CR_n(-\infty)=\Gamma$.

It is worth noting that the functions $H$ corresponding to these
distributions do indeed approach the form identified above. Thus, the
variational problem is not merely a relaxation: its optimal value is
exactly $1/\Gamma$, the inverse of the asymptotic competitive ratio.

\subsection{Proof of \Cref{prop:equalizing-rule}}\label{subsec:equalizing-rules-proof}

Fix a distribution $\D$  with mean $\mu$ and $n\ge2$.

Starting with $i=n-1$ and proceeding backward, choose $q_i\in(0,1)$
such that
\[
\frac{S(q_i)}{q_i}
= \mu
\prod_{j=i+1}^{n-1}q_j,
\]
where the empty product is $1$, and let $\pi^\star$ be the resulting
quantile-threshold rule.

We first verify that this recursion is well-defined. Suppose that
$q_{i+1},\ldots,q_{n-1}$ have already been chosen, and let
$t_i=\mu \prod_{j=i+1}^{n-1}q_j\in(0,\mu]$. The function $q\mapsto S(q)/q$
is continuous on $(0,1]$, tends to $+\infty$ as $q\downarrow0$, and
equals $0$ at $q=1$. The intermediate value theorem therefore gives
some $q_i\in(0,1)$ such that $S(q_i)/q_i=t_i$.

Moreover, every chosen threshold satisfies $q_i>F(0)$. This is
immediate if $F(0)=0$. Otherwise, for every $q\le F(0)$, we have
$S(q)=\mu$, and hence $S(q)/q=\mu/q\ge \mu/F(0)> \mu\ge t_i$. Thus no
solution to the defining equation can lie in $[0,F(0)]$.

Let $r_i=\prod_{j=1}^{i-1}q_j$ be the probability that the item
reaches agent $i$. For every $i<n$, the defining equation gives
$S(q_i)=\mu \prod_{j=i}^{n-1}q_j$, and therefore
\[
v_i(\pi^\star,\D)
=
r_iS(q_i)
=
\left(\prod_{j=1}^{i-1}q_j\right)
\left(\prod_{j=i}^{n-1}q_j\right) \cdot \mu 
=
\mu r_n.
\]
The final agent receives the item whenever reached, so
$v_n(\pi^\star,\D)=r_n\E[X]= \mu r_n$ as well. Thus $\pi^\star$ gives every agent
expected utility $ \mu r_n$.

Since $q_i>F(0)$ for every $i<n$, \Cref{prop:pareto} implies that
$\mathbf v(\pi^\star,\D)=(\mu r_n,\ldots,\mu r_n)$ is Pareto optimal. If
$\ALG_n(-\infty;\D)>\mu r_n$, then some online rule would give every agent
expected utility strictly greater than $\mu r_n$, contradicting this
Pareto optimality. Hence $\ALG_n(-\infty;\D)=\mu r_n$, and $\pi^\star$ is
optimal.\qed

\subsection{Proof of \Cref{lem:variational-reduction}}

Fix a distribution $\D$ satisfying $\E[X]=1$ and $n\ge 2$. Let
$q_1,\ldots,q_{n-1}$ be the thresholds of the optimal equalizing rule
from \Cref{prop:equalizing-rule}, and let $r_i$ be its reach
probabilities. Thus, $r_1=1$, $r_{i+1}=r_iq_i$, and
$c=\ALG_n(-\infty;\D)=r_n$. Because every $q_i\in(F(0),1)$, we have
$c\in(0,1)$.

Recall that $S(q)=\int_q^1F^{-1}(u)\,du$. Thus, $S$ is continuous and
nonincreasing on $[0,1]$, constant and equal to $1$ on $[0,F(0)]$,
and strictly decreasing from $1$ to $0$ on $[F(0),1]$. Consequently,
$S^{-1}$ is continuous and strictly decreasing on $[0,1]$, with
$S^{-1}(0)=1$ and $S^{-1}(1)=F(0)$.

We first verify that $H\in\K_{\mathrm{fin}}$. By construction, $H$ is
nonnegative, nondecreasing, satisfies $H(0)=0$, and equals $+\infty$
for every $z>1/c$. Moreover, $S$ is concave because $F^{-1}$ is
nondecreasing. Its inverse $S^{-1}$ is therefore concave and
decreasing. Since $\log$ is increasing and concave,
$\log S^{-1}$ is concave wherever $S^{-1}$ is positive. It follows
that $z\mapsto-(n-1)\log S^{-1}(cz)$ is convex on $[0,1/c]$.
Extending this function by $+\infty$ beyond $1/c$ preserves
convexity, so $H\in\K_{\mathrm{fin}}$.

We next compute $\I(H)$. Let
$U_{-1}^{\max}=\max\{U_2,\ldots,U_n\}$. Allocating the item to the
agent with the largest quantile implements the prophet's allocation
and treats the agents symmetrically. Conditioning on
$U_{-1}^{\max}$ gives
$\OPT_n(\D)=\E[S(U_{-1}^{\max})]$. Since
$0\le S(U_{-1}^{\max})\le1$, the tail-integral formula yields
\begin{align*}
\frac{\OPT_n(\D)}{c}
&=\int_0^{1/c}
  \Pr\left[S(U_{-1}^{\max})\ge cz\right]\,dz \\
&=\int_0^{1/c}
  \Pr\left[U_{-1}^{\max}\le S^{-1}(cz)\right]\,dz \\
&=\int_0^{1/c}\left(S^{-1}(cz)\right)^{n-1}\,dz \\
&= \int_0^{\infty} e^{-H(z)} \, dz
=\I(H).
\end{align*}

It remains to prove $\C(H) \le 1$. Because the equalizing rule gives
every agent expected utility $c$, we have $r_iS(q_i)=c$ for every
$i<n$. Since $q_i>F(0)$ and $S$ is strictly decreasing on
$[F(0),1]$, it follows that $q_i=S^{-1}(c/r_i)$ and hence
$H(1/r_i)=-(n-1)\log q_i$. Moreover,
$r_{i+1}=r_iq_i$, so
$\log(r_i/r_{i+1})=-\log q_i=H(1/r_i)/(n-1)$.

Because $H$ is nondecreasing, for every $i=1,\ldots,n-1$,
\[
\int_{1/r_i}^{1/r_{i+1}}\frac{dz}{zH(z)}
\le \frac{1}{H(1/r_i)} \int_{1/r_i}^{1/r_{i+1}}\frac{dz}{z}
= \frac{\log(r_i/r_{i+1})}{H(1/r_i)}
=
\frac1{n-1}.
\]
These intervals are consecutive and cover $[1,1/c]$, since $r_1=1$
and $r_n=c$. Summing over $i$, and using $H(z)=+\infty$ for
$z>1/c$, gives $\C(H)\le1$. This completes the proof of
\Cref{lem:variational-reduction}. \qed

\subsection{Proof of \Cref{lem:variational-bound}}\label{sec: proof of variational lemma}

Fix an arbitrary $G\in\K_{\mathrm{fin}}$ satisfying $\C(G)\le 1$. Choose $R>1$ sufficiently large, so that $G(z)=+\infty$ for every $z\ge R$.
Let $\F_R$ be the class of all nonnegative, nondecreasing, convex
functions $X\colon[0,\infty)\to[0,\infty]$ such that: (1) $X(0)=0$, (2) $X(z)=+\infty$  for every $z\ge R$, and $\C(X)\le 1$. In particular, $G\in\F_R$. Define
\[
    \Lambda_R\coloneqq\sup_{X\in\F_R}\I(X).
\]

The following lemma ensures that we may work with an actual maximizer
and records the basic properties needed for later arguments.
Its proof is deferred to \Cref{app:maximizer-existence}.

\begin{lemma}\label{lem:maximizer-existence}
For every \(R>1\), the supremum defining \(\Lambda_R\) is attained.
Moreover, \(\Lambda_R>1\), and every maximizer \(G^\star\in\F_R\)
satisfies
\(
\C(G^\star)=1.
\)
Defining
\[
M\coloneqq\sup\{z\ge0:G^\star(z)<+\infty\},
\qquad
J\coloneqq\sup\{z\ge0:G^\star(z)=0\},
\]
then
\[
0\le J<1<M\le R.
\]
\end{lemma}

Fix a maximizer $G^\star\in\F_R$, and let $J$ and $M$ be as in \Cref{lem:maximizer-existence}. We may assume that $G^\star(M)=+\infty$, since setting its value at $M$ to $+\infty$ preserves convexity and changes neither $\I(G^\star)$ nor $\C(G^\star)$.

\paragraph{Hinges and hinge perturbations.}

For $a\in[0,M)$, define the \emph{hinge} $h_a(z)\coloneqq(z-a)_+=\max\{z-a,0\}$. Hinges are the elementary building blocks of convex functions: their coefficients record where and by how much the slope increases. To characterize $G^\star$, we want to determine which hinges appear in its hinge representation. We begin by computing the first-order effect of adding a small hinge at each location $a$ (and then renormalize to preserve the constraint $\C=1$). The resulting first-order conditions will restrict which locations can have positive weight in the hinge representation of $G^\star$.

For $\varepsilon\ge0$, let $G_{a,\varepsilon} \coloneqq G^\star+\varepsilon h_a$. This function is nonnegative, nondecreasing, and convex, and it has the same finite cutoff $M$ as $G^\star$. Moreover, $G_{a,\varepsilon}\ge G^\star$, so $0<\C(G_{a,\varepsilon}) \le\C(G^\star) =1$.
Define the normalized perturbation
\[
    \widetilde G_{a,\varepsilon}
    \coloneqq
    \C(G_{a,\varepsilon})G_{a,\varepsilon}.
\]
Then $\widetilde G_{a,\varepsilon}\in\F_R$ and $\C(\widetilde G_{a,\varepsilon}) = \frac{\C(G_{a,\varepsilon})}
         {\C(G_{a,\varepsilon})}
    =1$. In particular, since $\C(G^\star)=1$, $\widetilde G_{a,0}=G^\star$.

We now compute the first-order effect of this perturbation. Since
$G^\star$ maximizes $\I$ over $\F_R$ and
$\widetilde G_{a,\varepsilon}\in\F_R$ for every $\varepsilon\ge0$, $\I(\widetilde G_{a,\varepsilon})
    \le
    \I(G^\star)$. Consequently, if the right derivative exists, then
\[
    \Psi(a)
    \coloneqq
    \left.
    \frac{d}{d\varepsilon}
    \I(\widetilde G_{a,\varepsilon})
    \right|_{\varepsilon=0+}
    \le0.
\]
The following lemma computes this derivative; its proof in
Appendix~\ref{app:first-variation} gives the formal details.

\begin{lemma}\label{lem:first-variation}
For $a\in[0,M)$, the right derivative $\Psi(a)$ exists. Define
\begin{align*}
    A(a)
    &\coloneqq
    -\left.
    \frac{d}{d\varepsilon}
    \C(G^\star+\varepsilon h_a)
    \right|_{\varepsilon=0+}
    =
    \int_1^M
    \frac{h_a(z)}{zG^\star(z)^2}\,dz,\\
    B(a)
    &\coloneqq
    -\left.
    \frac{d}{d\varepsilon}
    \I(G^\star+\varepsilon h_a)
    \right|_{\varepsilon=0+}
    =
    \int_0^M h_a(z)e^{-G^\star(z)}\,dz,\\
    \lambda
    &\coloneqq
    -\left.
    \frac{d}{d\varepsilon}
    \I\bigl((1+\varepsilon)G^\star\bigr)
    \right|_{\varepsilon=0}
    =
    \int_0^M G^\star(z)e^{-G^\star(z)}\,dz.
\end{align*}
Then $\Psi(a)=\lambda A(a)-B(a)$. All three quantities are finite and strictly positive.
\end{lemma}
We will also use the following derivatives of $\Psi$, derived in
Appendix~\ref{app:psi-derivatives}.
\begin{lemma}\label{lem:psi-derivatives}
The function $\Psi$ is continuously differentiable on $[0,M)$, with
\[
    \Psi'(a)
    =
    \begin{cases}
    \displaystyle
    -\lambda\int_1^M\frac{dz}{zG^\star(z)^2}
    +\int_a^M e^{-G^\star(z)}\,dz,
    & 0\le a\le1,\\[4mm]
    \displaystyle
    \int_a^M
    \left(
        e^{-G^\star(z)}
        -\frac{\lambda}{zG^\star(z)^2}
    \right)\,dz,
    & 1<a<M.
    \end{cases}
\]
Moreover, $\Psi$ is twice continuously differentiable away from
$a=1$, with
\[
    \Psi''(a)
    =
    \begin{cases}
        -e^{-G^\star(a)}, & 0<a<1,\\[2mm]
        \displaystyle
        \frac{\lambda}{aG^\star(a)^2}
        -e^{-G^\star(a)}, & 1<a<M.
    \end{cases}
\]
\end{lemma}

\paragraph{The hinge representation and its support.} 

We next show that $G^\star$ is linear on $(J,M)$.
We use the standard hinge representation of a convex function.
Because $G^\star$ is finite and convex on $[0,M)$, we have
$G^\star(J)=0$. There is a nonnegative, locally finite measure
$\mu^\star$ on $[J,M)$ such that
\begin{equation}\label{eq:hinge-representation}
    G^\star(z)
    =
    \int_{[J,M)}h_a(z)\,\mu^\star(da)
    \qquad(0\le z<M).
\end{equation}

Here the atom of $\mu^\star$ at $J$ equals the initial right slope $G_+^{\star\prime}(J)$, while its mass after $J$ records subsequent increases in the slope. The following lemma, proved in
Appendix~\ref{app:lem:averaged-first-order-condition}, shows that the
aggregate first-order effect of the hinges appearing in $G^\star$ is
zero. Intuitively, if these hinges had a strictly negative aggregate
effect, slightly subtracting them from $G^\star$ would preserve
convexity and yield a first-order improvement. Furthermore, perturbing $G^\star$ in the aggregate direction of all its
constituent hinges simply rescales $G^\star$, an effect that is undone
by the normalization. Its aggregate first-order effect must therefore
be zero.

\begin{lemma}\label{lem:averaged-first-order-condition}
For the measure $\mu^\star$ in
\Cref{eq:hinge-representation}, $\int_{[J,M)}\Psi(a)\,\mu^\star(da)=0$.
\end{lemma}

Since $\Psi(a)\le0$ for every $a$, \Cref{lem:averaged-first-order-condition} implies that $\mu^\star\bigl(\{a\in[J,M):\Psi(a)<0\}\bigr)=0$. Recall that the support of $\mu^\star$ is the set
\[
    \operatorname{supp}(\mu^\star)
    \coloneqq
    \left\{
        a\in[J,M):
        \mu^\star\bigl(
            (a-\delta,a+\delta)\cap[J,M)
        \bigr)>0
        \text{ for every }\delta>0
    \right\}.
\]
In other words, a point belongs to the support if every neighborhood of that point receives positive $\mu^\star$-measure. We claim that $\Psi$ vanishes everywhere on this support. Indeed,
suppose that
$a_0\in\operatorname{supp}(\mu^\star)$ but $\Psi(a_0)<0$.
By continuity, there are $\delta,\eta>0$ such that $\Psi(a)\le-\eta$ whenever $a\in(a_0-\delta,a_0+\delta)\cap[J,M)$.
This neighborhood has positive $\mu^\star$-measure because
$a_0$ lies in the support. It would therefore make a strictly negative
contribution to the integral of $\Psi$. Since $\Psi\le0$ everywhere,
this would imply
\[
    \int_{[J,M)}\Psi(a)\,\mu^\star(da)<0,
\]
contradicting the equality above. Hence
\begin{equation}\label{eq:psi-zero-on-support}
    \Psi(a)=0
    \qquad
    \text{for every }a\in\operatorname{supp}(\mu^\star).
\end{equation}

Moreover, $G^\star(z)>0$ for every $z>J$. Thus
\Cref{eq:hinge-representation} implies $\mu^\star([J,z))>0$ for every $z>J$. Therefore, $J\in\operatorname{supp}(\mu^\star)$, and hence $\Psi(J)=0$.

We have shown that $J\in\operatorname{supp}(\mu^\star)$ and that $\Psi$ vanishes on $\operatorname{supp}(\mu^\star)$. To prove that $G^\star$ consists of a single hinge, it therefore remains to show that $\mu^\star$ has no other support points. For $(J, 1]$, this is relatively straightforward:  By
\Cref{lem:psi-derivatives}, $ \Psi''(a)=-e^{-G^\star(a)}<0$ for $J<a<1$. Thus $\Psi$ is strictly concave on $[J,1]$. Since $\Psi\le0$ and $\Psi(J)=0$, strict concavity gives $\Psi(a)<0$ for every $a\in(J,1]$: otherwise, $\Psi$ would be positive between $J$ and $a$. Together with \Cref{eq:psi-zero-on-support}, this implies $\operatorname{supp}(\mu^\star)\cap(J,1]=\varnothing$.
The range $(1, M)$ is more involved, and we defer the details to Appendix~\ref{app:empty-support}.

\begin{lemma}\label{lem:empty-support}
    $\mathrm{supp}(\mu^\star) \cap (J, M) = \varnothing$.
\end{lemma}

The above arguments combined show that
$\operatorname{supp}(\mu^\star)=\{J\}$.
Thus $\mu^\star=\sigma\delta_J$ for some $\sigma>0$, and,
\begin{equation}\label{eq:optimal-hinge-form}
    G^\star(z)
    =
    \begin{cases}
        0, & 0\le z\le J,\\
        \sigma(z-J), & J<z<M,\\
        +\infty, & z\ge M.
    \end{cases}
\end{equation}

\paragraph{Evaluating the single-hinge maximizer.}

It remains to evaluate this function. First suppose that $J>0$, and
define $
    a\coloneqq\sigma J$ and
    $b\coloneqq\sigma(M-J)$.
Then $a,b>0$. Since $\C(G^\star)=1$,
\[
    1 =
    \frac1{\sigma}
    \int_1^M\frac{dz}{z(z-J)} =
    \frac1{\sigma J}
    \log\left(
        \frac{M-J}{M(1-J)}
    \right) =
    \frac1a
    \log\left(
        \frac{b}{(a+b)(1-J)}
    \right).
\]
Therefore, $J = 1-\frac{b}{a+b}e^{-a}$. Since $1/\sigma=J/a$, we get
\begin{align*}
    \I(G^\star)
    &=
    J+\frac{1-e^{-b}}{\sigma}\\
    &=
    \left(
        1-\frac{b}{a+b}e^{-a}
    \right)
    \left(
        1+\frac{1-e^{-b}}a
    \right)\\
    &=
    \frac1{\rho(a,b)}
    \le\frac1\Gamma.
\end{align*}

If $J=0$, then $\C(G^\star)=1$ gives
\(
    \sigma=1-\frac1M.
\)
Setting $b=\sigma M=M-1>0$, we get $\I(G^\star) = \frac{1-e^{-b}}{\sigma}
    =
    \frac{b+1}{b}(1-e^{-b})$.
Here, the parameter $a=\sigma J$ from the preceding case equals zero, which lies outside the domain $a>0$ in the definition of $\rho$. However, this case is obtained by taking $a\downarrow0$: $\I(G^\star)
    = \lim_{a\downarrow0}\frac1{\rho(a,b)}$; and hence again $\I(G^\star)\le\frac1\Gamma$.

Finally, the original function $G$ belongs to $\F_R$, so
\(
    \I(G)
    \le
    \I(G^\star)
    \le
    \frac1\Gamma.
\)
This proves \Cref{lem:variational-bound}. \qed

\subsection{Proof of \Cref{thm:main}}

\subsubsection{Lower bound: $\CR_n(-\infty) \ge \Gamma$}

Fix any distribution $\D$ with finite, positive mean, and $n \ge 2$. Scaling all
values by a positive constant scales both
$\ALG_n(-\infty;\D)$ and $\OPT_n(\D)$ by that constant, and therefore
does not change their ratio. We may thus assume that $\E[X]=1$.

By \Cref{lem:variational-reduction}, there is a function
$H\in\K_{\mathrm{fin}}$ satisfying
\[
\frac{\OPT_n(\D)}{\ALG_n(-\infty;\D)}
=
\I(H)
\qquad\text{and}\qquad
\C(H)\le1.
\]
Applying \Cref{lem:variational-bound} gives
\[
\frac{\OPT_n(\D)}{\ALG_n(-\infty;\D)}
=
\I(H)
\le
\frac1\Gamma.
\]
Equivalently,
\[
\frac{\ALG_n(-\infty;\D)}{\OPT_n(\D)}
\ge\Gamma.
\]
Since this holds for every distribution $\D$ with finite, positive
mean, taking the infimum over $\D$ proves
\(
\CR_n(-\infty)\ge\Gamma.
\)

\subsubsection{Upper bound: $\limsup_{n \to \infty} \CR_n(-\infty)\leq\Gamma$}\label{subsubsec: hard distribution}

Next, we prove our upper bound on $\CR_n(-\infty)$.

Fix arbitrary $a,b>0$. We construct a sequence of distributions whose competitive ratios converge to at most $\rho(a,b)$ (as defined in the beginning of~\Cref{sec: results}). Fix $n>b$, and pick a sufficiently small $\eps>0$, so that $\eps <1 - \frac bn  $ and $\frac{1}{\varepsilon} > \frac na$. Let $\D_{n,\eps}$ take value $1/\eps$ with probability $\eps$, value
$n/a$ with probability $b/n$, and value $0$ otherwise. Its mean is $\E[X] = \frac1\eps\eps+\frac na\frac bn = 1+\frac ba$. The constraint on $\varepsilon$ ensures that this is well-defined, and that $1/\varepsilon$ truly is the highest value. Abbreviating $S_{\D_{n,\eps}}$ to $S$, we have
\[
S(q)=
\begin{cases}
\displaystyle
1+\frac ba,
  &0\le q\le1-\eps-b/n,\\[3mm]
\displaystyle
1+\frac na(1-\eps-q),
  &1-\eps-b/n<q\le1-\eps,\\[3mm]
\displaystyle
\frac{1-q}{\eps},
  &1-\eps<q\le1.
\end{cases}
\]
We obtain the desired bound by first letting $\eps\downarrow0$ and then
letting $n\to\infty$. The argument below computes the online and
prophet values directly, although the construction can also be viewed
as realizing the maximizers from the optimization problem.\footnote{To see this, write $\mu=1+b/a$. After taking
$\eps\downarrow0$, the recurrence below becomes an equality and gives
$c_n=\left(1-\frac{b}{a+b}(1+a/n)^{-(n-1)}\right)^{-1}\to c_\infty$,
where
$c_\infty=\left(1-\frac{b}{a+b}e^{-a}\right)^{-1}$.
Moreover, $S^{-1}(y)=1$ for $y\le1$ and
$S^{-1}(y)=1-\frac an(y-1)$ for $1<y\le\mu$. Consequently, after
normalizing the distribution by $\mu$, the associated functions $H$
converge almost everywhere to the function that is zero
on $[0,J]$, has slope $\sigma=ac_\infty$ on $(J,M)$, and is infinite
for $z>M$, where $J=1/c_\infty$ and $M=\mu/c_\infty$. In particular,
$\sigma J=a$ and $\sigma(M-J)=b$. This exactly gives the objective value $1/\rho(a, b)$.}

\paragraph{Bounding $\ALG_n$.}
Let $c=\ALG_n(-\infty;\D_{n,\eps})$. Let $q_1, \ldots, q_{n-1}$ be the thresholds of the rule in \Cref{prop:equalizing-rule} and let $r_i$ be the probability that the item is still available
when agent $i$ arrives. Thus $r_1=1$, $r_{i+1}=r_iq_i$, and $r_iS(q_i)=c$ for every $i<n$.
The last agent's utility is $c=r_n\E[X] = r_n\left(1+\frac ba\right)$,
so $r_n=\frac{ac}{a+b}$.

We claim that
\[
  c
  \le
  \frac{1}{
    1-\dfrac{b}{a+b}
      \left(\dfrac n{n+a}\right)^{n-1}
  }.
\]
The right hand side is greater than $1$, so the claim is immediate if $c\le1$. It remains to consider $c>1$.

Since $r_i\le1$, we have $S(q_i)=\frac{c}{r_i}>1$ for every $i<n$. Moreover, the optimal thresholds satisfy $q_i>F(0)=1-\eps-\frac bn$. Since $S(q)\le1$ for $q\ge1-\eps$, the inequality $S(q_i)>1$ implies $q_i<1-\eps$. Together with $q_i>1-\eps-b/n$, this shows that $1-\eps-\frac bn<q_i<1-\eps$. Substituting $S(q_i)=1+\frac na(1-\eps-q_i)$ into $r_iS(q_i)=c$ gives $r_i(1-\eps-q_i)= \frac an(c-r_i)$. It follows that
\begin{align*}
  c-r_{i+1}
  &=c-r_iq_i\\
  &=\left(1+\frac an\right)(c-r_i)+\eps r_i\\
  &\ge
  \left(1+\frac an\right)(c-r_i).
\end{align*}
Since $r_1=1$, by induction we have $c-r_n \ge
  \left(1+\frac an\right)^{n-1}(c-1)$. On the other hand, $c-r_n = c-\frac{ac}{a+b} = c\frac{b}{a+b}$. Therefore,
\[
  \left(1+\frac an\right)^{n-1}(c-1)
  \le
  c\frac{b}{a+b}.
\]
Rearranging gives
\[
  c
  \le
  \frac{1}{
    1-\dfrac{b}{a+b}
      \left(1+\dfrac an\right)^{-(n-1)}
  }
  =
  \frac{1}{
    1-\dfrac{b}{a+b}
      \left(\dfrac n{n+a}\right)^{n-1}
  }.
\]

\paragraph{Bounding $\OPT_n$.}
The maximum equals $1/\eps$ if at least one draw takes the high value. This event has probability $1-(1-\eps)^n$. If no draw takes the high value but at least one draw takes the middle value, the maximum equals $n/a$. This event has probability $(1-\eps)^n - \left(1-\eps-\frac bn\right)^n$.
Therefore
\[
  \OPT_n(\D_{n,\eps})
  =
  \frac1n\left[
    \frac{1-(1-\eps)^n}{\eps}
    +
    \frac na
    \left(
      (1-\eps)^n
      -
      \left(1-\eps-\frac bn\right)^n
    \right)
  \right].
\]
Taking $\eps \rightarrow 0$ gives
\[
  \lim_{\eps \rightarrow 0^+}
  \OPT_n(\D_{n,\eps})
  =
  1+\frac1a
  \left[
    1-\left(1-\frac bn\right)^n
  \right].
\]

\paragraph{Putting it all together.}
Since $\CR_n(-\infty)$ is the infimum over distributions, combining the
online and offline bounds gives
\[
  \CR_n(-\infty)
  \le
  \frac{1}{
    \left(
      1-\dfrac{b}{a+b}
        \left(\dfrac n{n+a}\right)^{n-1}
    \right)
    \left(
      1+\dfrac1a
        \left[
          1-\left(1-\dfrac bn\right)^n
        \right]
    \right)
  }.
\]

Finally, $\left(\frac n{n+a}\right)^{n-1} = \left(1+\frac an\right)^{-(n-1)} \rightarrow e^{-a}$, 
and $\left(1-\frac bn\right)^n \rightarrow e^{-b}$. It follows that
\[
  \limsup_{n\to\infty}\CR_n(-\infty)
  \le
  \frac{1}{
    \left(
      1-\dfrac{b}{a+b}e^{-a}
    \right)
    \left(
      1+\dfrac{1-e^{-b}}a
    \right)
  }
  = \rho(a, b).
\]

Because $a,b>0$ were arbitrary,
\begin{align*}
  \limsup_{n\to\infty}\CR_n(-\infty) \le  \inf_{a, b > 0} \rho(a, b) = \Gamma.
\end{align*}
Combining this upper bound with
$\CR_n(-\infty)\ge\Gamma$ for every $n$ completes the proof of
\Cref{thm:main}. \qed

\section*{Acknowledgments}
Alexandros Psomas was supported in part by an NSF CAREER award CCF-2144208, and a grant by the AI Security Institute (AISI) via the Alignment Project. Abhiram Manohara was supported in part by an NSF CAREER award CCF-2144208.
\paragraph{AI Disclosure} 
We worked on this problem for several months and obtained weaker bounds before using GPT-5.6 Sol. Its initial attempts did not yield the tight bound, but over extended interactions it proposed recasting the remaining problem as an optimization problem, and suggested a variational inequality. In subsequent exchanges, it assisted in developing the statement and proof of the current version of this inequality (\Cref{lem:variational-bound}). We then verified the proof and developed the exposition, with particular emphasis on readability and intuition. The authors take full responsibility for the correctness of the results and all other content in the paper.


\bibliographystyle{alpha}
\bibliography{refs}

\appendix

\section{Proof of \Cref{prop:small-n}}\label{subsec: proof of gap for 3 agents}

First consider \(n=2\), and normalize so that \(\E[X]=1\). If the first agent uses quantile threshold \(q\), then the two agents receive expected utilities \(S(q)\) and \(q\). The optimal egalitarian threshold therefore satisfies \(S(q)=q\), and hence \(\ALG_2(-\infty;\D)=q\) and \(S^{-1}(q)=q\). Moreover, concavity of \(S\), together with \(S(0)=1\) and \(S(1)=0\), implies \(S(q)\ge 1-q\), so \(q\ge 1/2\).

The prophet's expected utility per agent is
\(\OPT_2(\D)=\int_0^1 S(t)\,dt=\int_0^1 S^{-1}(s)\,ds\), where the second equality follows by integrating the region under \(S\) in the opposite order. Let \(-a\) be a supergradient of the concave function \(S^{-1}\) at \(q\). Then \(S^{-1}(s)\le \min\{1,q-a(s-q)\}\). The endpoint conditions \(S^{-1}(0)=1\) and \(S^{-1}(1)\ge0\) further imply \((1-q)/q\le a\le q/(1-q)\). Therefore,
\begin{align*}
\OPT_2(\D)
&\le \int_0^1\min\{1,q-a(s-q)\}\,ds \\
&=1-\frac{(1-q)^2}{2}\left(a+2+\frac1a\right)
 \le 1-2(1-q)^2,
\end{align*}
where the last inequality follows from \(a+1/a\ge2\). Therefore,
\[
 \frac{\ALG_2(-\infty;\D)}{\OPT_2(\D)}
\ge \frac{q}{1-2(1-q)^2}
\ge \frac{2+\sqrt2}{4},
\]
where the last inequality follows by minimizing over \(q\in[1/2,1]\). Since $\D$ was arbitrary, $\CR_2(-\infty) \ge (2 + \sqrt2) / 4$. Hill and Kertz show that the classical two-agent ratio satisfies
\(\CR_2(1)=(2+\sqrt2)/4\)~\cite{hillkertz1982}. Monotonicity in \(p\) now gives the claimed equality for every \(p\in[-\infty,1]\).

For \(n=3\), apply the hard-distribution construction (\Cref{subsubsec: hard distribution}) with \(a=1\) and \(b=3/2\). Equivalently, for sufficiently small \(\eps>0\), consider the distribution that places probabilities \(\eps\), \(1/2\), and \(1/2-\eps\) on the values \(1/\eps\), \(3\), and \(0\), respectively. The bounds from that construction give
\[
\CR_3(-\infty)
\le
\frac{1}{
 \left(1-\frac35\left(\frac34\right)^2\right)
 \left(1+\left[1-\left(\frac12\right)^3\right]\right)}
=\frac{128}{159}
\approx 0.80503.
\]
By comparison, the sharp classical three-agent ratio is
\(\CR_3(1)\approx 1/1.22138\approx0.81875\)~\cite{hillkertz1982}. Thus
\(\CR_3(-\infty)<\CR_3(1)\), completing the proof. \qed


\section{Proof of \Cref{prop:pareto}}\label{app:pareto}

Let $R_i$ be the event that the item is still available when agent $i$ arrives, and let $r_i=\Pr(R_i)$ be the probability that this event occurs. For an online rule $\pi$, for each $i<n$ with $r_i>0$, let $\alpha_i=\Pr(\pi=i\mid R_i)$ and set $q_i=1-\alpha_i$. If $r_i=0$, then agent $i$ and all later agents have zero utility, and the remaining thresholds may be chosen arbitrarily.

Since $R_i$ depends only on previous values and the rule's internal randomness, it is independent of $U_i$. Thus, conditional on $R_i$, the quantile $U_i$ remains uniform on $[0,1]$. Averaging over the histories that lead to $R_i$, the rule accepts agent $i$ on a total quantile mass $\alpha_i=1-q_i$.
Among all such acceptance rules, accepting
the upper interval $[q_i,1]$ maximizes the expected utility of $i$. Therefore, $v_i(\pi,\D)\le r_i S(q_i)$.

The item reaches agent $i+1$ precisely when it reaches agent $i$ and $\pi$ does not allocate to agent $i$. Therefore, $r_{i+1}=r_i(1-\alpha_i)=r_iq_i$. A threshold rule $\pi'$ with thresholds $q_1,\ldots,q_{n-1}$ also rejects agent $i$, conditional on reaching them, with probability $q_i$. Since both $\pi$ and $\pi'$ reach agent $1$ with probability $1$, the preceding recurrence shows inductively that $\pi$ and $\pi'$ have the same reach probability $r_i$ for every agent. Therefore, $v_i(\pi',\D)=r_iS(q_i)\ge v_i(\pi,\D)$ for all $i < n$. Furthermore, since $\pi'$ always allocates to the final agent when reached, $v_n(\pi',\D) = r_n\E[X] \geq v_n(\pi,\D)$. This proves the first claim.

For the second claim, fix a threshold rule $\pi$ with thresholds
$q_i\ge F(0)$. Suppose that some online rule Pareto dominates $\pi$. By the first claim, we may assume that the dominating rule is also a
threshold rule, with thresholds $\widetilde q_1,\ldots,\widetilde q_{n-1}$. We may also assume that $\widetilde q_i\ge F(0)$: replacing a threshold below $F(0)$ by $F(0)$
only rejects zero-valued realizations and cannot decrease any agent's
utility. Let $r_i=\prod_{j<i}q_j$, and $\widetilde r_i=\prod_{j<i}\widetilde q_j$
be the two rules' reach probabilities. 

We will show by induction that $\widetilde r_i\le r_i$ for every
$i\in[n]$. The statement is trivial for $i=1$, since $\widetilde r_1=r_1=1$. Suppose that $\widetilde r_i\le r_i$ (and $i < n$). If $r_i=0$, then $\widetilde r_i=0$, and hence $\widetilde r_{i+1}=r_{i+1}=0$. Now suppose that $r_i>0$. If $q_i<1$, then $S(q_i)>0$, so Pareto
domination implies $\widetilde r_i>0$ and $S(\widetilde q_i) \ge \frac{r_i}{\widetilde r_i}S(q_i) \ge S(q_i)$. Since $S$ is strictly decreasing on $[F(0),1]$, it follows that $\widetilde q_i\le q_i$. If $q_i=1$, then $\widetilde q_i\le q_i$ holds trivially. In either case, $\widetilde r_{i+1} =\widetilde r_i\widetilde q_i \le r_iq_i  =r_{i+1}$, completing the induction.

We now show that Pareto domination forces the two threshold rules to
have identical utility profiles. First, suppose that $q_i>0$ for every $i\in[n-1]$. The preceding induction also gives $\widetilde q_i\le q_i$ for every $i\in[n-1]$. Thus, whenever $\widetilde r_i<r_i$, we have $\widetilde r_{i+1}\le\widetilde r_iq_i<r_iq_i=r_{i+1}$. Any strict inequality therefore persists until agent $n$, giving $v_n(\widetilde\pi,\D)=\widetilde r_n\E[X]<r_n\E[X]=v_n(\pi,\D)$, contrary to Pareto domination. Hence $ \widetilde r_i=r_i$ for all $i \in [n]$. It follows (by the definition of $\widetilde r_i$ and $r_i$) that $\widetilde q_i=q_i$ for every $i\in[n-1]$, and therefore the two threshold rules have the same utility profile.

Finally, suppose that $q_k=0$, and let $k$ be the first such index. Then $F(0)=0$ and $S(q_k)=S(0)=\E[X]$. Pareto domination requires that $\widetilde r_kS(\widetilde q_k) \ge r_k\E[X]$. Since $\widetilde r_k\le r_k$ and $S(\widetilde q_k)\le\E[X]$, equality must hold throughout. Hence $\widetilde r_k=r_k$ and $S(\widetilde q_k)=\E[X]$. Since $F(0)=0$ and $S$ is strictly decreasing on $[0,1]$, this
implies $\widetilde q_k=0=q_k$. Moreover, if any earlier inequality $\widetilde r_i\le r_i$ were strict, it would persist until agent $k$,
since $q_i>0$ for every $i<k$. This would contradict $\widetilde r_k=r_k$. Thus the two rules agree through agent $k$, and both have zero reach probability afterward. Their utility profiles again coincide. Therefore no online rule can strictly Pareto dominate $\pi$. \qed

\section{Proofs for the Fair Division Results}\label{app:fair-division}

\subsection{Proof of~\Cref{prop:fair-division-ex-ante}}
The key idea of this proof is showing that for every allocation rule for the optimization problem, there is a rule for the fair division problem that achieves exactly $m$ times the expected value for each agent, and for every allocation rule for the fair division problem, there is a rule for the optimization problem that achieves exactly $\frac{1}{m}$ times the expected value for each agent. Further, for any online allocation rule for either problem, the corresponding rule for the other problem is also online. 

Fix $n,m$, and $\D$. 
Given a single-item allocation rule $\pi$,
apply it independently to every item. Each item contributes
$v_i(\pi,\D)$ to agent $i$'s expected utility, so linearity of
expectation gives total expected utility $mv_i(\pi,\D)$.
This construction preserves online feasibility.

Conversely, fix a many-item allocation rule, and let $A_i$ denote
the random bundle it allocates to agent $i$. Construct a single-item
rule $\pi$ as follows. Choose an item $K$ uniformly from $[m]$,
independently of all values and internal randomness. Use the real
single-item values $X_1,\ldots,X_n$ as the values for item $K$, and
independently sample the values for all other items from $\D$.
Simulate the many-item rule on this valuation matrix, and allocate
the real item to whichever agent receives item $K$, leaving it
unallocated if the simulated rule does so.

Conditional on any choice of $K$, the simulated valuation matrix
has the same distribution as the original many-item instance.
Therefore,
\[
v_i(\pi,\D)
=
\frac1m\sum_{j=1}^m
\E\bigl[X_{ij}\ind_{j\in A_i}\bigr]
=
\frac{\E[U_i]}m.
\]
If the many-item rule is online, supply the simulated values as
the corresponding agents arrive. Its decision at agent $i$ then
uses only the real values observed so far and internally generated
values, so the resulting single-item rule is also online.
The construction applies to offline rules as well.

Thus, the online and offline expected-utility sets in the many-item
problem are exactly $m$ times their single-item counterparts.
Since $w_p$ is positively homogeneous (scaling all entries by $\lambda$ scales the output by $\lambda$), it follows that
\(
\ALG^{\mathrm{ea}}_{n,m}(p;\D)=m\ALG_n(p;\D)\) and $
\OPT^{\mathrm{ea}}_{n,m}(p;\D)=m\OPT_n(\D)$.
Taking ratios and then the infimum over $\D$ gives
$\CR^{\mathrm{ea}}_{n,m}(p)=\CR_n(p)$. \qed

\subsection{Proof of~\Cref{prop:fair-division-impossibility}}
Fix $p\le0$ and $n\ge 2$, and consider the distribution
$\D=\operatorname{Bernoulli}(\lambda)$, where
$\lambda=2\log(n)/n$. If any agent has zero utility, the $p$-mean
is zero. Thus, positive welfare requires every agent to receive
at least one value-one item. Since  $m=n$, this means that every agent receives
exactly one item, which they value at one. Consequently, the
$p$-mean is either zero or one, and its expectation equals the
probability that every agent receives a value-one item.

\paragraph{Prophet.}

The prophet can give every agent a value-one item if there is
a permutation $\pi$ of the items such that $X_{i,\pi(i)}=1$
for every agent $i$. Define a bipartite graph $G$ whose two parts
represent the agents and items, with an edge between agent $i$
and item $j$ if and only if $X_{ij}=1$. Finding such an allocation
is equivalent to finding a perfect matching in $G$.

Since the valuations are independent Bernoulli random variables,
$G$ is a random bipartite graph $G(n,n,\lambda)$. For
$\lambda=2\log(n)/n$, this graph contains a perfect matching
with probability $1-o(1)$~\cite{frieze2015introduction}.
Therefore, $\OPT^{\mathrm{ep}}_{n,n}(p;\D)=1-o(1)$.

\paragraph{Online algorithm.}

Fix any online allocation rule and condition on the history
before the final agent arrives. Positive welfare is possible
only if each of the first $n-1$ agents has received exactly
one value-one item and exactly one item remains. On any such
history, the remaining item's value to the final agent is an
independent Bernoulli random variable with parameter $\lambda$.
Thus, the conditional probability of positive welfare is at
most $\lambda$. On every other history, it is zero.

It follows that every online rule has expected welfare at most
$\lambda$, so $\ALG^{\mathrm{ep}}_{n,n}(p;\D)\le\lambda$.
Therefore,
\[
\CR^{\mathrm{ep}}_{n,n}(p)
\le
\frac{\ALG^{\mathrm{ep}}_{n,n}(p;\D)}
     {\OPT^{\mathrm{ep}}_{n,n}(p;\D)}
\le
\frac{2\log(n)}{n(1-o(1))}
=
O\left(\frac{\log n}{n}\right).
\]
In particular, $\CR^{\mathrm{ep}}_{n,n}(p)\to0$ as
$n\to\infty$ for every $p\le0$. \qed

\subsection{Proof of \Cref{prop:fair-division-limit}}
\label{app:prop:fair-division-limit}


Fix a distribution $\D$, $n \ge 2$ and $p \in [-\infty, 1]$.
Fix a single-item allocation rule $\pi$, and apply it independently
to every item, using independent internal randomness for each copy.
Write agent $i$'s total utility as $U_i=\sum_{j=1}^m Y_{ij}$, where
$Y_{ij}$ is their utility from item $j$. For each agent $i$, the
variables $Y_{ij}$ are i.i.d.\ across items, with expectation
$v_i(\pi,\D)$ and finite mean since $0\le Y_{ij}\le X_{ij}$.
By the strong law of large numbers,
$U_i/m\to v_i(\pi,\D)$ almost surely. Since there are finitely many
agents and the $p$-mean is continuous on nonnegative utility
profiles, this gives
\(
w_p(\mathbf U/m)\longrightarrow w_p(\mathbf v(\pi,\D))\) almost surely.

We next show that the expectations also converge. Since the
$p$-mean is nonnegative, Fatou's lemma gives
\[
w_p(\mathbf v(\pi,\D))
\le
\liminf_{m\to\infty}\E[w_p(\mathbf U/m)].
\]
On the other hand, concavity of the $p$-mean and Jensen's
inequality give
$\E[w_p(\mathbf U/m)]\le
w_p(\E[\mathbf U/m])=w_p(\mathbf v(\pi,\D))$
for every $m$. Combining these bounds and using positive
homogeneity, we obtain
\[
\lim_{m\to\infty}\frac{\E[w_p(\mathbf U)]}{m}
=
w_p(\mathbf v(\pi,\D)).
\]

This provides a lower bound on the optimal ex-post welfare.
Indeed, for every online single-item rule $\pi$, independent
repetition is a feasible online many-item rule, so
\[
\liminf_{m\to\infty}
\frac{\ALG^{\mathrm{ep}}_{n,m}(p;\D)}m
\ge w_p(\mathbf v(\pi,\D)).
\]
Taking the supremum over online single-item rules gives a lower
bound of $\ALG_n(p;\D)$. The same argument for offline rules
gives a lower bound of $\OPT_n(\D)$.

For the upper bounds, Jensen's inequality shows that the
ex-post welfare of any rule is at most its ex-ante welfare.
Applying \Cref{prop:fair-division-ex-ante}, we therefore have
\[
\ALG^{\mathrm{ep}}_{n,m}(p;\D)\le m\ALG_n(p;\D),
\qquad
\OPT^{\mathrm{ep}}_{n,m}(p;\D)\le m\OPT_n(\D).
\]
Together, the upper and lower bounds show that
\[
\lim_{m\to\infty}\frac{\ALG^{\mathrm{ep}}_{n,m}(p;\D)}m
=\ALG_n(p;\D),
\qquad
\lim_{m\to\infty}\frac{\OPT^{\mathrm{ep}}_{n,m}(p;\D)}m
=\OPT_n(\D).
\]
Since $\OPT_n(\D)>0$, taking the ratio gives
\[
\lim_{m\to\infty}
\frac{\ALG^{\mathrm{ep}}_{n,m}(p;\D)}
     {\OPT^{\mathrm{ep}}_{n,m}(p;\D)}
=
\frac{\ALG_n(p;\D)}{\OPT_n(\D)}.
\]
This completes the proof. \qed

\subsection{Finite Convergence Bounds for Restricted Distributions}\label{app:bounds}

We now give quantitative bounds on the number of items needed
for $p$-mean welfare to approach the single-item
guarantee.

\begin{proposition}
\label{prop:fair-division-finite-convergence}
Fix $n\ge2$, $p\in[-\infty,1]$, $\delta\in(0,1)$, and a
distribution $\D$ with mean $\mu$.
Each of the following conditions is sufficient to guarantee
\[
\frac{\ALG^{\mathrm{ep}}_{n,m}(p;\D)}
     {\OPT^{\mathrm{ep}}_{n,m}(p;\D)}
\ge
\frac{\ALG_n(p;\D)}{\OPT_n(\D)}-\delta.
\]
\begin{enumerate}[leftmargin=*]
\item $\D$ has finite variance $\sigma^2$, and
\[
m\ge
\frac{9n^3}{\delta^2}
\left(1+\frac{\sigma^2}{\mu^2}\right).
\]
\item $X\le B\mu$ almost surely, and
\[
m\ge
\frac{24nB}{\delta^2}\log\left(\frac{2n}{\delta}\right).
\]
\end{enumerate}
\end{proposition}

\begin{proof}
For $p=1$, the ex-ante and ex-post objectives coincide, so
the result follows from \Cref{prop:fair-division-ex-ante}.
Henceforth, assume $p<1$.

Fix an optimal single-item rule for the $p$-mean objective,
and let $\mathbf v=(v_1,\ldots,v_n)$ be its expected-utility
profile. Such a rule exists by \Cref{prop:pareto}, since
threshold profiles depend continuously on thresholds in
the compact set $[0,1]^{n-1}$. For $p=-\infty$, choose the
equalizing rule from \Cref{prop:equalizing-rule}.

Write $a=w_p(\mathbf v)=\ALG_n(p;\D)$ and
$V=\OPT_n(\D)$. Applying the rule independently to each item,
including independent internal randomness, gives total
utilities $U_i=\sum_{j=1}^m Y_{ij}$, where the $Y_{ij}$
are independent across items, satisfy $0\le Y_{ij}\le X_{ij}$,
and have expectation $v_i$. Let $Z_i=U_i/m$.

By Jensen's inequality and
\Cref{prop:fair-division-ex-ante},
$\OPT^{\mathrm{ep}}_{n,m}(p;\D)\le mV$. Consequently,
\[
\frac{\ALG^{\mathrm{ep}}_{n,m}(p;\D)}
     {\OPT^{\mathrm{ep}}_{n,m}(p;\D)}
\ge
\frac{\E[w_p(\mathbf Z)]}{V}.
\]
Since $a/V\le1$, it suffices to prove
$\E[w_p(\mathbf Z)]\ge(1-\delta)a$.

\paragraph{Properties of an optimal single-item rule.}

Allocating to a uniformly random agent gives every agent
expected utility $\mu/n$, while every agent's expected utility
is at most $\mu$. Thus, $\mu/n\le a\le\mu$.

For finite $p<1$, every $v_i$ is positive. For $p\le0$,
this follows from $a>0$. For $0<p<1$, maximizing $w_p$ is equivalent to maximizing $\sum_i v_i^p$. If some coordinate is zero, mixing toward $(\mu/n,\ldots,\mu/n)$ with probability $\varepsilon$ increases the contribution of the zero coordinates by $\Omega(\varepsilon^p)$, while changing the contribution of the positive coordinates by at most $O(\varepsilon)$. Since $\varepsilon^p/\varepsilon\to\infty$ as $\varepsilon\downarrow0$, this contradicts optimality.

Define
\[
g_i=\frac{\partial w_p}{\partial v_i}(\mathbf v)
=\frac1n a^{1-p}v_i^{p-1}.
\]
This formula also holds for $p=0$, and
$\sum_i g_iv_i=a$. Allocating the item to agent $i$
regardless of its value gives the feasible profile
$\mu\mathbf e_i$. By optimality of $\mathbf v$, mixing
toward this profile cannot have a positive directional
derivative. Therefore, $
\mu g_i\le a$ and hence $
v_i^{1-p}\ge\frac{\mu}{na^p}$.

Let $p_0=1/(1+\log n)$. If $p\le0$, the preceding bound
and $a\ge\mu/n$ give $v_i\ge\mu/n$. If $0<p\le p_0$,
using $a\le\mu$ instead gives
\[
v_i\ge\mu n^{-1/(1-p)}
\ge\frac{\mu}{en}
\ge\frac{\mu}{3n}.
\]
For $p=-\infty$, the equalizing rule also satisfies
$v_i=a\ge\mu/n$. Thus, whenever $p\le p_0$, every agent
has expected utility at least $\mu/(3n)$.

For $p_0<p<1$, we will instead use the following bound,
valid for every nonnegative profile $\mathbf z$:
\[
a-w_p(\mathbf z)
\le
\frac1p\sum_i g_i(v_i-z_i)_+.
\]
To verify it, let $d_i=(v_i-z_i)_+$ and
$t=\sum_i g_id_i/a\in[0,1]$. Since
$(1-s)^p\ge1-s$ for $s\in[0,1]$,
\[
w_p(\mathbf z)^p
\ge
\frac1n\sum_i(v_i-d_i)^p
\ge a^p(1-t).
\]
Taking the $p$th root and using
$(1-t)^{1/p}\ge1-t/p$ proves the bound.

\paragraph{Finite variance.}

Suppose that $\D$ has finite variance $\sigma^2$, and
write $M_2=\E[X^2]=\sigma^2+\mu^2$. Since
$0\le Y_{ij}\le X_{ij}$ and the contributions are
independent across items,
$\operatorname{Var}(Z_i)\le M_2/m$.

First suppose that $p\le p_0$. Define
$T=\max_i(v_i-Z_i)_+/v_i$. Then $T\in[0,1]$ and
$\mathbf Z\ge(1-T)\mathbf v$ coordinatewise, so
$w_p(\mathbf Z)\ge(1-T)a$. Moreover,
\[
\E[T]
\le
\sqrt{\sum_i
\frac{\E[(v_i-Z_i)_+^2]}{v_i^2}}
\le
\sqrt{\sum_i\frac{\operatorname{Var}(Z_i)}{v_i^2}}
\le
3\sqrt{\frac{n^3M_2}{m\mu^2}}.
\]
Here the first inequality follows from bounding the maximum
by the Euclidean norm and applying Jensen's inequality.

Now suppose that $p_0<p<1$. Using the bound from the
preceding paragraph, $g_i\le a/\mu$, and
$\E[(v_i-Z_i)_+]\le\sqrt{\operatorname{Var}(Z_i)}$,
we obtain
\[
a-\E[w_p(\mathbf Z)]
\le
\frac1p\sum_i g_i\E[(v_i-Z_i)_+]
\le
\frac{na}{p\mu}\sqrt{\frac{M_2}{m}}
\le
2a\sqrt{\frac{n^3M_2}{m\mu^2}},
\]
where the last inequality uses
$1/p<1+\log n\le2\sqrt n$.

Thus, in either case,
\[
\E[w_p(\mathbf Z)]
\ge
a\left(
1-3\sqrt{\frac{n^3}{m}
\left(1+\frac{\sigma^2}{\mu^2}\right)}
\right).
\]
The stated bound on $m$ makes the error at most $\delta$.

\paragraph{Bounded values.}

Now suppose that $X\le B\mu$ almost surely. First consider
$p\le p_0$. Since $Y_{ij}\in[0,B\mu]$ and
$v_i\ge\mu/(3n)$, a multiplicative Chernoff bound gives
\[
\Pr\left[Z_i<(1-\delta/2)v_i\right]
\le
\exp\left(-\frac{\delta^2mv_i}{8B\mu}\right)
\le
\exp\left(-\frac{\delta^2m}{24nB}\right).
\]
Under the stated bound on $m$, a union bound shows that
$\mathbf Z\ge(1-\delta/2)\mathbf v$ coordinate wise with
probability at least $1-\delta/2$. Nonnegativity,
monotonicity, and positive homogeneity therefore give
\[
\E[w_p(\mathbf Z)]
\ge(1-\delta/2)^2a
\ge(1-\delta)a.
\]

It remains to consider $p_0<p<1$. We use a moment bound
to control the expected welfare directly. For any
nonnegative random variable $Z$ with mean $v>0$ and
finite second moment, H\"older's inequality gives
\[
\E[Z^p]
\ge
\frac{v^{2-p}}{(\E[Z^2])^{1-p}}
=
v^p\left(1+\frac{\operatorname{Var}(Z)}{v^2}\right)^{-(1-p)}
\ge
v^p-(1-p)v^{p-2}\operatorname{Var}(Z).
\]
The last inequality follows from convexity of
$t\mapsto(1+t)^{-(1-p)}$.

Since $Y_{ij}^2\le B\mu Y_{ij}$,
$\operatorname{Var}(Z_i)\le B\mu v_i/m$.
Applying the moment bound and using
$\sum_i g_i\le na/\mu$, we obtain
\[
\begin{aligned}
\frac1n\sum_i\E[Z_i^p]
&\ge
a^p-\frac{(1-p)B\mu}{nm}\sum_i v_i^{p-1}\\
&=
a^p-\frac{(1-p)B\mu}{m}a^{p-1}\sum_i g_i\\
&\ge
a^p\left(1-\frac{(1-p)nB}{m}\right).
\end{aligned}
\]
The assumed bound on $m$ implies $m\ge nB$. Since
$x\mapsto x^{1/p}$ is convex, Jensen's inequality gives
\[
\begin{aligned}
\E[w_p(\mathbf Z)]
&\ge
\left(\frac1n\sum_i\E[Z_i^p]\right)^{1/p}\\
&\ge
a\left(1-\frac{(1-p)nB}{pm}\right)\\
&\ge
a\left(1-\frac{nB\log n}{m}\right),
\end{aligned}
\]
where the final inequality uses $(1-p)/p<\log n$.
The stated bound on $m$ makes this error at most $\delta$,
completing the proof.
\end{proof}

\section{Technical Details for the Proof of \Cref{lem:variational-bound}}
\subsection{Proof of \Cref{lem:maximizer-existence}}\label{app:maximizer-existence}
Observe that $\Lambda_R>1$. Indeed, define
\[
    L_R(z)=
    \begin{cases}
        \left(1-\dfrac1R\right)z, & 0\le z<R,\\[2mm]
        +\infty, & z\ge R.
    \end{cases}
\]
Then $\C(L_R) = \frac{1}{1-1/R}\int_1^R\frac{dz}{z^2} =1$, whereas
\[
    \I(L_R)
    =
    \int_0^R e^{-(1-1/R)z}\,dz
    =
    \frac{1-e^{-(R-1)}}{1-1/R}
    >1.
\]
The final inequality follows from $e^{R-1}>R$.

\paragraph{Existence of a maximizer.}

We first show that the supremum defining $\Lambda_R$ is attained.
Let $(G_i)$ be a sequence in $\F_R$ such that $\I(G_i)\longrightarrow\Lambda_R$, and define $\varphi\colon[0,\infty]\to[0,1]$, as $\varphi(x)=\frac{x}{1+x}$, with $\varphi(+\infty)=1$. Let
\[
    G'_i=\varphi\circ G_i.
\]
Each $G'_i$ is nondecreasing and takes values in $[0,1]$.  We use the following form of Helly's selection theorem.

\begin{lemma}[Helly's selection theorem {\cite[Lemma 13.15]{realanalysis}}]
\label{lem:helly}
Let $(f_i)$ be a sequence of nondecreasing functions
$f_i\colon [0, T]\to[0,1]$. Then there exists a nondecreasing function
$f\colon [0, T] \to[0,1]$ and a subsequence $(f_{i_j})$ such that $\lim_{j \to \infty} f_{i_j}(x) = f(x)$ for every $x\in[0, T]$.
\end{lemma}

Applying Helly's selection theorem (\Cref{lem:helly}) to $(G'_i)$ on $[0,R]$, there exist indices
$i_1<i_2<\cdots$ and a nondecreasing function $G'\colon[0,R]\to[0,1]$ such that
$G'_{i_j}(z)\longrightarrow G'(z)$ for every $z\in[0,R]$. Since $G'_i(z)=1$ for every $z\ge R$, we extend $G'$ by setting
$G'(z)=1$ for $z\ge R$.

Define
\[
    G^\star(z)=\varphi^{-1}(G'(z))
    =\frac{G'(z)}{1-G'(z)},
\]
where the right-hand side is interpreted as $+\infty$ when $G'(z)=1$.
Then $G_{i_j}(z)\longrightarrow G^\star(z)$
for every $z\ge0$. Pointwise limits preserve nonnegativity, monotonicity, and convexity. Moreover, $G^\star(0)=0$, and since $G_{i_j}(z)=+\infty$ for every
$z\ge R$, we also have $G^\star(z)=+\infty$ for every $z\ge R$.
Fatou’s lemma allows us to upper-bound the integral of the pointwise limit by the liminf of the integrals:
\[
    \C(G^\star) =\int_1^\infty
        \lim_{j \to\infty}\frac{1}{zG_{i_j}(z)}\,dz \le
        \liminf_{j \to\infty}
        \int_1^\infty\frac{dz}{zG_{i_j}(z)}
    \le 1.
\]
Thus \(G^\star\in\F_R\).

It remains to verify that the objective value is preserved in the
limit. We have that $0\le e^{-G_{i_j}(z)}\le \ind\{z<R\}$, and the function on the right is integrable. Therefore, the dominated
convergence theorem gives
\[
\I(G^\star) =\int_0^\infty e^{-G^\star(z)}\,dz =\lim_{j \to\infty}\int_0^\infty e^{-G_{i_j}(z)}\,dz =\Lambda_R.
\]
Hence \(G^\star\) is a maximizer over \(\F_R\), and since the original
function \(G\) belongs to \(\F_R\), $\I(G)\le\I(G^\star)$.

\paragraph{Basic properties of the maximizer.} Define $M\coloneqq\sup\{z\ge0:G^\star(z)<+\infty\}$ and $J\coloneqq\sup\{z\ge0:G^\star(z)=0\}$. Since $G^\star\in\F_R$, we have $M\le R$. 

First, $M>1$. Indeed, we showed above that $\I(G^\star)=\Lambda_R>1$, so if $M\le1$, then
\[
    \I(G^\star)
    =\int_0^M e^{-G^\star(z)}\,dz
    \le M
    \le1,
\]
which is a contradiction.

We next show that $J<1$. Suppose instead that $J\ge1$. Since
$G^\star$ is nonnegative and nondecreasing, this implies
$G^\star(1)=0$. Choose any $z_0\in(1,M)$. If
$G^\star(z_0)=0$, then $\C(G^\star)=+\infty$. Otherwise, convexity
gives, for $1<z<z_0$, $G^\star(z) \le \frac{z-1}{z_0-1}G^\star(z_0)$. Consequently,
\[
\C(G^\star) \ge \int_1^{z_0}\frac{dz}{zG^\star(z)} \ge \frac{z_0-1}{G^\star(z_0)} \int_1^{z_0}\frac{dz}{z(z-1)} =+\infty,
\]
contradicting $\C(G^\star)\le1$. Hence $0\le J<1<M\le R$.

We also claim that $\C(G^\star)=1$. Since $J<1<M$, the function $G^\star$ is finite and strictly
positive on $(1,M)$, so $\C(G^\star)>0$. Suppose that $c\coloneqq\C(G^\star)<1$. Then $cG^\star\in\F_R$ and $\C(cG^\star) =\frac{\C(G^\star)}{c} =1$. Moreover, since $0<c<1$ and $G^\star(z)>0$ for $1<z<M$, $e^{-cG^\star(z)}>e^{-G^\star(z)}$. It follows that $\I(cG^\star)>\I(G^\star)$, contradicting the maximality of $G^\star$ over $\F_R$. Therefore
$\C(G^\star)=1$. \qed

\subsection{Proof of \Cref{lem:first-variation}}\label{app:first-variation}

Fix \(a\in[0,M)\), abbreviate \(G^\star\) to \(G\), and write
\[
h=h_a,
\qquad
G_\varepsilon=G+\varepsilon h,
\qquad
c(\varepsilon)=\C(G_\varepsilon).
\]
Recall that \(J<1<M\), \(G\) is finite on \([0,M)\), and
\(\C(G)=1\). In particular, \(G(1)>0\), and monotonicity gives
\(G(z)\ge G(1)\) for \(z\ge1\).

We first justify the derivative of the constraint. Since
\(G_\varepsilon=+\infty\) on \((M,\infty)\), for every
\(\varepsilon>0\),
\[
\frac{c(\varepsilon)-c(0)}{\varepsilon}
=
-\int_1^M
\frac{h(z)}
{zG(z)\bigl(G(z)+\varepsilon h(z)\bigr)}
\,dz.
\]
The integrand converges pointwise as \(\varepsilon\downarrow0\) to
\(-h(z)/(zG(z)^2)\). Moreover,
\[
0\le
\frac{h(z)}
{zG(z)\bigl(G(z)+\varepsilon h(z)\bigr)}
\le
\frac{h(z)}{zG(z)^2}
\le
\frac{M}{G(1)}\frac{1}{zG(z)}.
\]
The final function is integrable because \(\C(G)=1\). Dominated
convergence therefore gives
\[
c'(0+)
=
-\int_1^M\frac{h(z)}{zG(z)^2}\,dz
=
-A(a).
\]

We next differentiate the unnormalized objective. For
\(\varepsilon>0\),
\[
\frac{\I(G_\varepsilon)-\I(G)}{\varepsilon}
=
-\int_0^M
e^{-G(z)}
\frac{1-e^{-\varepsilon h(z)}}{\varepsilon}
\,dz.
\]
Since \(0\le 1-e^{-t}\le t\) for \(t\ge0\), the absolute value of the
integrand is bounded by
\[
h(z)e^{-G(z)}\le M\ind\{z<M\},
\]
which is integrable. Dominated convergence thus yields
\[
\left.\frac{d}{d\varepsilon}\I(G+\varepsilon h)
\right|_{\varepsilon=0+}
=
-\int_0^M h(z)e^{-G(z)}\,dz
=
-B(a).
\]

The same argument applies to the scaling perturbation. For
\(|\varepsilon|\le1/2\), the derivative of
\(e^{-(1+\varepsilon)G(z)}\) is bounded in absolute value by
\(G(z)e^{-G(z)/2}\), which is uniformly bounded for \(G(z)\ge0\).
Hence differentiation under the integral sign is valid and gives
\[
\left.
\frac{d}{d\varepsilon}
\I\bigl((1+\varepsilon)G\bigr)
\right|_{\varepsilon=0}
=
-\int_0^M G(z)e^{-G(z)}\,dz
=
-\lambda.
\]

It remains to differentiate the normalized perturbation. By definition,
\[
\widetilde G_{a,\varepsilon}
=
c(\varepsilon)G_\varepsilon
=
c(\varepsilon)\bigl(G+\varepsilon h\bigr).
\]
Because \(c(0)=1\) and \(c'(0+)=-A(a)\), for every \(z<M\),
\[
\frac{\widetilde G_{a,\varepsilon}(z)-G(z)}{\varepsilon}
\longrightarrow
h(z)-A(a)G(z).
\]
We verify that this pointwise derivative may be passed through the
objective integral. Since \(c(\varepsilon)\to1\), for all sufficiently
small \(\varepsilon>0\) we have \(1/2\le c(\varepsilon)\le2\), and
\((c(\varepsilon)-1)/\varepsilon\) is bounded. The mean value theorem
therefore gives
\[
\left|
\frac{
e^{-\widetilde G_{a,\varepsilon}(z)}-e^{-G(z)}
}{\varepsilon}
\right|
\le
\bigl(KG(z)+2M\bigr)e^{-G(z)/2}
\]
for some constant \(K\) independent of \(z\) and \(\varepsilon\).
The right-hand side is bounded on \([0,M)\), since both
\(x e^{-x/2}\) and \(e^{-x/2}\) are bounded on \([0,\infty)\), and it
vanishes outside the finite interval \([0,M)\). It is therefore an
integrable dominating function. Dominated convergence now gives
\begin{align*}
\Psi(a)
&=
\left.
\frac{d}{d\varepsilon}
\I(\widetilde G_{a,\varepsilon})
\right|_{\varepsilon=0+} \\
&=
-\int_0^M
\bigl(h(z)-A(a)G(z)\bigr)e^{-G(z)}\,dz \\
&=
\lambda A(a)-B(a).
\end{align*}

Finally, all three quantities are finite and strictly positive. Indeed,
\[
A(a)
\le
\frac{M}{G(1)}
\int_1^M\frac{dz}{zG(z)}
=
\frac{M}{G(1)}
<\infty,
\qquad
B(a)\le M^2<\infty,
\]
and \(\lambda<\infty\) because \(xe^{-x}\le1/e\) for \(x\ge0\).
Because \(a<M\) and \(M>1\), the integrand defining \(A(a)\) is
strictly positive on
\((\max\{1,a\},M)\), while the integrand defining \(B(a)\) is strictly
positive on \((a,M)\). Finally, \(G(z)>0\) on \((J,M)\), so the
integrand defining \(\lambda\) is strictly positive there. Hence
\(A(a),B(a),\lambda>0\). \qed

\subsection{Proof of \Cref{lem:psi-derivatives}}\label{app:psi-derivatives}

Abbreviate \(G^\star\) to \(G\), and define
\[
f(z)\coloneqq\frac{1}{zG(z)^2}
\quad (1<z<M),
\qquad
g(z)\coloneqq e^{-G(z)}
\quad (0<z<M).
\]
Both functions are integrable on their respective domains. Indeed,
\(0\le g\le1\), while \(G(z)\ge G(1)>0\) for \(z\ge1\), and hence
\[
\int_1^M f(z)\,dz
\le
\frac{1}{G(1)}
\int_1^M\frac{dz}{zG(z)}
=
\frac{1}{G(1)}
<\infty.
\]
Moreover, because \(G\) is finite and convex on \([0,M)\), it is
continuous on \((0,M)\). Thus \(g\) is continuous on \((0,M)\), and
\(f\) is continuous on \((1,M)\).

We will repeatedly use the following elementary differentiation
formula. If \(w\) is integrable and nonnegative on \((\ell,M)\), then
\[
K_w(a)\coloneqq\int_a^M(z-a)w(z)\,dz
\]
satisfies
\[
K_w'(a)=-\int_a^M w(z)\,dz.
\]
To see this without invoking an unproved Leibniz rule, Tonelli's theorem
gives
\[
K_w(a)
=
\int_a^M\int_a^z w(z)\,dt\,dz
=
\int_a^M\left(\int_t^M w(z)\,dz\right)dt.
\]
The tail integral \(t\mapsto\int_t^M w(z)\,dz\) is continuous, so the
fundamental theorem of calculus gives the asserted derivative. If
\(w\) is continuous at \(a\), differentiating the tail integral once
more gives
\[
K_w''(a)=w(a).
\]

Apply this formula first to
\[
B(a)=\int_a^M(z-a)g(z)\,dz.
\]
For every \(a\in[0,M)\), where the derivative at \(a=0\) is understood
as a right derivative,
\[
B'(a)=-\int_a^M e^{-G(z)}\,dz.
\]
This derivative is continuous in \(a\). Since \(g\) is continuous on
\((0,M)\), we also have
\[
B''(a)=e^{-G(a)}
\qquad (0<a<M).
\]

We next differentiate \(A\). If \(0\le a\le1\), then \(z\ge1\) implies
\(h_a(z)=z-a\), so
\[
A(a)=\int_1^M(z-a)f(z)\,dz
\]
is affine in \(a\), with
\[
A'(a)=-\int_1^M f(z)\,dz
=
-\int_1^M\frac{dz}{zG(z)^2}.
\]
In particular, \(A''(a)=0\) for \(0<a<1\).

If \(1<a<M\), then
\[
A(a)=\int_a^M(z-a)f(z)\,dz.
\]
The preceding differentiation formula gives
\[
A'(a)
=
-\int_a^M f(z)\,dz
=
-\int_a^M\frac{dz}{zG(z)^2},
\]
and, since \(f\) is continuous on \((1,M)\),
\[
A''(a)=f(a)=\frac{1}{aG(a)^2}.
\]
The two formulas for \(A'\) agree at \(a=1\), because
\[
\lim_{a\downarrow1}
\left(-\int_a^M f(z)\,dz\right)
=
-\int_1^M f(z)\,dz.
\]
Thus \(A\) is continuously differentiable on \([0,M)\), with a
one-sided derivative at \(0\), and is twice continuously
differentiable separately on \((0,1)\) and \((1,M)\).

By \Cref{lem:first-variation},
\(\Psi(a)=\lambda A(a)-B(a)\). Combining the preceding derivative
formulas gives
\[
\Psi'(a)
=
\begin{cases}
\displaystyle
-\lambda\int_1^M\frac{dz}{zG(z)^2}
+\int_a^M e^{-G(z)}\,dz,
&0\le a\le1,\\[4mm]
\displaystyle
\int_a^M
\left(
e^{-G(z)}
-\frac{\lambda}{zG(z)^2}
\right)dz,
&1<a<M.
\end{cases}
\]
The expressions agree at \(a=1\), and the continuity of tail integrals
shows that \(\Psi'\) is continuous on \([0,M)\). Therefore \(\Psi\) is
continuously differentiable there.

Finally, differentiating once more on either side of \(a=1\) yields
\[
\Psi''(a)
=
\begin{cases}
-e^{-G(a)}, &0<a<1,\\[2mm]
\displaystyle
\frac{\lambda}{aG(a)^2}-e^{-G(a)},
&1<a<M.
\end{cases}
\]
Because \(G\) is continuous on \((0,M)\), these expressions are
continuous on \((0,1)\) and \((1,M)\), respectively. Hence \(\Psi\) is
twice continuously differentiable away from \(a=1\), as claimed. \qed

\subsection{Proof of~\Cref{lem:averaged-first-order-condition}}\label{app:lem:averaged-first-order-condition}

We average the first-order effect $\Psi(a)$ over all the hinges
appearing in the representation of $G^\star$. The key point is that
averaging these hinges against $\mu^\star$ reproduces $G^\star$ itself.

We begin with $A$. Recalling its definition and applying Tonelli's
theorem to the nonnegative integrand, we obtain
\begin{align*}
    \int_{[J,M)}A(a)\,\mu^\star(da)
    &=
    \int_{[J,M)}
    \left(
        \int_1^M
        \frac{h_a(z)}{zG^\star(z)^2}\,dz
    \right)\mu^\star(da)\\
    &=
    \int_1^M
    \frac{
        \int_{[J,M)}h_a(z)\,\mu^\star(da)
    }{zG^\star(z)^2}\,dz\\
    &=
    \int_1^M
    \frac{G^\star(z)}{zG^\star(z)^2}\,dz\\
    &=
    \int_1^M\frac{dz}{zG^\star(z)}\\
    &=
    \C(G^\star)
    =1.
\end{align*}

The analogous calculation for $B$ gives
\begin{align*}
    \int_{[J,M)}B(a)\,\mu^\star(da)
    &=
    \int_{[J,M)}
    \left(
        \int_0^M
        h_a(z)e^{-G^\star(z)}\,dz
    \right)\mu^\star(da)\\
    &=
    \int_0^M
    \left(
        \int_{[J,M)}h_a(z)\,\mu^\star(da)
    \right)
    e^{-G^\star(z)}\,dz\\
    &=
    \int_0^M
    G^\star(z)e^{-G^\star(z)}\,dz\\
    &=
    \lambda.
\end{align*}
Again, Tonelli's theorem applies because every integrand being
exchanged is nonnegative.

Since $\Psi(a)=\lambda A(a)-B(a)$, these two identities imply
\begin{align*}
    \int_{[J,M)}\Psi(a)\,\mu^\star(da)
    &=
    \lambda
    \int_{[J,M)}A(a)\,\mu^\star(da)
    -
    \int_{[J,M)}B(a)\,\mu^\star(da)\\
    &=\lambda\C(G^\star)-\lambda\\
    &=0. \qedhere
\end{align*}

\subsection{Proof of \Cref{lem:empty-support}}\label{app:empty-support}

Abbreviate $G^\star$ to $G$. It remains to exclude support in
$(1,M)$. Define $m(z)=zG(z)^2e^{-G(z)}$. By
\Cref{lem:first-variation,lem:psi-derivatives}, for $1<a,z<M$,
\[
\Psi''(z)=\frac{\lambda-m(z)}{zG(z)^2},
\qquad
\Psi(a)=\int_a^M(z-a)\Psi''(z)\,dz.
\]

We first obtain a boundary condition on $m$. Since $G$ is
nondecreasing, the limit $L=\lim_{z\uparrow M}G(z)$ exists
in $(0,\infty]$. Consequently,
$m_M=\lim_{z\uparrow M}m(z)$ also exists, with $m_M=0$
if $L=+\infty$. If $m_M<\lambda$, then $\Psi''(z)>0$
for all $z$ sufficiently close to $M$. The integral identity
above would therefore give $\Psi(a)>0$ for $a$ sufficiently
close to $M$, contradicting $\Psi\le0$. Hence
$m_M\ge\lambda$, and in particular $L<\infty$.

The key observation is that $m$ has at most one peak.
On the interval where $G(z)\le2$, it is strictly increasing,
since $x\mapsto x^2e^{-x}$ is nondecreasing on $(0,2]$.
On the interval where $G(z)\ge2$, we have
\[
\log m(z)=\log z+2\log G(z)-G(z).
\]
The function $x\mapsto2\log x-x$ is concave and
nonincreasing on $[2,\infty)$, so its composition with
the convex function $G$ is concave. Adding $\log z$
makes $\log m$ strictly concave on this interval.
Thus, if $m(a)\ge\lambda$ for some $a\in(1,M)$,
then $m(z)>\lambda$ for every $a<z<M$: this follows
from strict increase up to $G=2$ and strict log-concavity
thereafter, using $m_M\ge\lambda$ at the right endpoint.

Now suppose that $a\in\operatorname{supp}(\mu^\star)
\cap(1,M)$. By \Cref{eq:psi-zero-on-support},
$\Psi(a)=0$. Since $\Psi\le0$, this is a local maximum,
so $\Psi''(a)\le0$ and hence $m(a)\ge\lambda$.
The preceding observation gives $m(z)>\lambda$,
and therefore $\Psi''(z)<0$, for every $a<z<M$.
But the integral identity then gives $\Psi(a)<0$,
a contradiction. Thus
$\operatorname{supp}(\mu^\star)\cap(1,M)=\varnothing$,
completing the proof.
 \qed

\end{document}